\documentclass[final,leqno,onefignum,onetabnum]{siamltex1213}
\usepackage{mathtools}
\usepackage{amsmath,amssymb}
\usepackage{esint}
\newtheorem{problem}{Problem}

\usepackage{array}
\begin{document}

\title{Random Multi-Hopper Model. Super-Fast Random Walks on Graphs}

\author{Ernesto Estrada\footnotemark[2], Jean-Charles Delvenne\footnotemark[3] \footnotemark[4], Naomichi
Hatano\footnotemark[5], Jos\'{e} L. Mateos\footnotemark[6], Ralf Metzler\footnotemark[7], Alejandro
P. Riascos\footnotemark[8], Michael T. Schaub\footnotemark[3] \footnotemark[4] \footnotemark[9]}

\maketitle

\global\long\def\thefootnote{\fnsymbol{footnote}}

\footnotetext[2]{Department of Mathematics \& Statistics, University of Strathclyde,
26 Richmond Street, Glasgow G11HQ, UK }

\footnotetext[3]{Department of Mathematical Engineering, Universit\'{e} catholique
de Louvain, 4 Avenue Georges Lem\^{a}\i tre, B-1348 Louvain-la-Neuve,
Belgium }

\footnotetext[4]{D\'{e}partement de Math\'{e}matique, University of Namur, 8 Rempart de la Vierge, B-5000 Namur, Belgium }

\footnotetext[5]{Institute of Industrial Science, University of Tokyo, Komaba,
Meguro, Tokyo 153-8505, Japan}

\footnotetext[6]{Instituto de F\i sica, Universidad Nacional Aut\'{o}noma de M\'{e}xico,
Apartado Postal 20-364, 01000 M\'{e}xico, Ciudad de M\'{e}xico, M\'{e}xico}

\footnotetext[7]{Institute for Physics \& Astronomy, University of Potsdam,
14476 Potsdam-Golm, Germany }

\footnotetext[8]{Department of Civil Engineering, Universidad Mariana, San
Juan de Pasto, Colombia}

\footnotetext[9]{present address: Institute for Data, Systems, and Society, Massachusetts 
Institute of Technology, Cambridge, USA}

\global\long\def\thefootnote{\arabic{footnote}}

\begin{abstract}
We develop a model for a random walker with long-range hops on general
graphs. This random multi-hopper jumps from a node to any other node in
the graph with a probability that decays as a function of the
shortest-path distance between the two nodes. We consider here two
decaying functions in the form of the Laplace and Mellin transforms of
the shortest-path distances. Remarkably, when the parameters of these
transforms approach zero asymptotically, the multi-hopper's hitting
times between any two nodes in the graph converge to their minimum
possible value, given by the hitting times of a normal random walker on
a complete graph. Stated differently, for small parameter values the
multi-hopper explores a general graph as fast as possible when compared
to a random walker on a full graph.
Using computational experiments we show that compared to the normal
random walker, the multi-hopper indeed explores graphs with clusters or
skewed degree distributions more efficiently for a large parameter
range. We provide further computational evidence of the speed-up
attained by the random multi-hopper model with respect to the normal
random walker by studying deterministic, random and real-world networks.
\end{abstract}

\section{Introduction}

Few mathematical models have found so many applications in
the physical, chemical, biological, social and economical sciences
as the random walk model~\cite{klafter2011first,masuda2016random}. The term ``random
walk'' was first proposed by K.~Pearson~\cite{pearson1905problem}
in an informal question posted in Nature in 1905.
Pearsons description reads as ``\textit{A man starts from a point $O$ and walks $l$ yards
in a straight line; he then turns through any angle whatever and walks
another $l$ yards in a second straight line. He repeats this process
$n$ times}.'' Among the first respondents to Pearson, 
Lord Rayleigh~\cite{rayleigh1905problem} 
already pointed out a connection existing
between random walks and other physical processes, namely that
a random walk ``\textit{is the same as that of the composition of
$n$ iso-periodic vibrations of unit amplitude and of phases distributed
at random}''. Connections like these between random walks and many physico-chemical,
biological and socio-economic processes are what guarantees the great
vitality of this research topic. This includes relations between random
walks, Brownian motion and diffusive processes in general~\cite{domb1978random},
the connection between random walks and the classical theory of 
electricity~\cite{doyle1984random,nash1959random}, and
the formulation of the efficient
market hypothesis~\cite{malkiel1991efficient}, among others. In
order to distinguish the originally-proposed random walk model from
its many varieties discussed in the literature we will
use the term normal random walk (NRW).

Many real world complex systems are more faithfully
represented as a network than as continuous system. This includes
ecological and biomolecular, social and economical as well as infrastructural
and technological networks~\cite{estrada2012structure}. The use
of random walk models in these systems also provides a large variety
of possibilities ranging from the analysis of the diffusion of information
and navigability on these networks to the exploration of their structures
to detect their fine-grained organization~\cite{noh2004random}.
From a mathematical point of view, these networks are nothing but (weighted)
graphs. We will use both terms interchangeably here with
preference to the term network for the case they are representing
some real-world system. 

The first work exploring the use of random
walks on graphs is credited to P\'{o}lya in 1921 when he was walking in
a park and he crossed the same couple very often~\cite{polya1921arithmetische}.
Thus, he asked the important question of the recurrence of a random
walker in an infinite regular lattice. Known today as \textit{P\'{o}lya's
recurrence theorem}, it states that a simple random walk on a $d$-dimensional
lattice is recurrent for $d=1,2$ and transient for $d>2$. Random
walks have been studied also on (regular) lattices, where many papers
in the physics literature have used this method for the study of diffusion
of atoms on a lattice. The classical papers of E.~W. Montroll in the
1960's and 1970's paved the way for the study of these phenomena~\cite{montroll1965random,montroll1969random,montroll1973random}. 

In the mid 1990's the group of G.~Ehrlich~\cite{senft1995long}
observed experimentally the self-diffusion of weakly bounded Pd atoms and made an interesting observation:
there were significant contributions to the thermodynamical
properties of the system from jumps spanning second and third nearest-neighbors
in the metallic surface, which can be considered
as a regular lattice. In 1997 the group of F.~Besenbacher~\cite{linderoth1997surface}
observed experimentally that for the self-diffusion of Pt atoms on a
Pt(110) surface, the jumps from non-nearest neighbors also contribute to the diffusion. Even more surprising
are the results of the same group, when they studied the diffusion of two large
organic molecules on a Cu(110) surface. In this case,
using scanning tunneling microscopy, they observed that long jumps
play a dominating role in the diffusion of the two organic molecules,
with root-mean-square jump lengths as large as $3.9$ and $6.8$ lattice
spacings~\cite{schunack2002long}. Since then the role of long-jumps
in adatom and molecules diffusing on metallic surface has been both
theoretically and experimentally confirmed in many different 
systems~\cite{yu2013single,ala2002collective}. Due to these experimental
evidence there has been some attempts to consider long-range jumps
in the diffusion of a particle on a regular lattice. The first of
them was the paper entitled \textit{``Lattice walks by long jumps}''
by Wrigley, Twigg and Ehrlich~\cite{wrigley1990lattice}.
Other works
have considered the space in which the diffusion takes place
to be continuous and applied the random-walk model with L\'{e}vy flights
to model these long-range effects (see below). However, the development
of a general multi-hopper model, in which a random walker hops to
any node of a general graph with probabilities depending on the
distance separating the corresponding nodes is still missing in the
literature. Apart from the physical scenarios related to the diffusion
of adatoms and admolecules on metallic surface, long-range jumps on
graphs are of a general interest. For instance, in social networks
one can take advantage of the full or partial knowledge of the network
beyond first acquaintances to diffuse information in a swifter way
than can be done by the traditional nearest-neighbor only strategy.
In exploring technological and infrastructural networks we can exploit
our knowledge of the topology of the network to jump from a position
to non-nearest-neighbors in such a way that we reach vaster regions
of the system explored in shorter times. 

In 2012 two groups published independently models that are designed
to account for all potential long jumps that a random walker can take
on a graph. Mateos and Riascos~\cite{riascos2012long}
proposed a random walk model where jumps occur to non-nearest
neighbors with a probability that decays as a power-law of the shortest-path
distance separating the two nodes (all formal definitions are done
in the Preliminaries section of this paper).  
Estrada~\cite{estrada2012path} generalized the concept of graph Laplacian
by introducing the $k$-path Laplacians, and used it to study generalized diffusion equations
to assess the influence of long-range
jumps on graphs. While the first paper by Mateos and Riascos provides
a probabilistic approach to the problem, the latter one by Estrada provides the
algebraic tools needed for its generalization and mathematical formalization. 

The combination of short- and long-range jumps in random walks is usually considered 
in the continuous space. This process is known as L\'{e}vy flight (LF)~\cite{shlesinger1995levy}.
LFs are widely
used to model efficient search processes, for instance, of animals
searching for sparse food. Due to the scale-free
nature of the underlying distribution of jump lengths,
the combination of more local search and occasional long, decorrelating excursions leads to less oversampling
and thus reduces typical search times, prompting the L{\'{e}}vy foraging
hypothesis~\cite{viswanathan2011physics}. This efficiency may, however,
be compromised in the presence of an external bias~\cite{palyulin2014levy}.
The efficiency of LFs in the random search context is, of course,
also a function of the dimensionality of the embedding space. In dimensions
three and above, when regular random walks become transient, LFs gradually lose
their advantage as an efficient strategy.  
LFs are Markovian random walks. At every jump the step length is drawn from a long-tailed probability
density function with the power-law decay~\cite{metzler2000random,metzler2004restaurant,hughes1996random}.
The Markovian requirement means that at every jump all memory to previous
jumps is erased, and thus the jump lengths $x$ are independent and
identically distributed random variables~\cite{bouchaud1990anomalous}.
What happens when the Markovian character of the random walk is broken?
A concrete example is an effective LF along a long polymer chain when
shortcuts are allowed at locations where the polymer loops back to
itself. The analogy with graphs runs in the following way. A long
polymer can be considered as an infinite path graph---a graph in which
all the nodes have degree two but two nodes are of degree one. Then, when
the polymer folds, some regions approach closely to others forming
a loop. In the graph this corresponds to the creation of cycles in
the graph. In the polymer it has been observed that as the length
stored in such loops has a power-law distribution, in the coordinate
system of the arc length along the polymer the shortcuts effectively
lead to a jump process with long tail~\cite{lomholt2005optimal,sokolov1997paradoxal}.
When the chain configurations relax quickly and, after taking a shortcut,
the random walker cannot use the same shortcut again, individual jumps
are indeed independent and thus the process is a true LF. However,
when the chain configurations are much slower compared to the rate
of shortcut events, correlations between present and previous jumps
become important: the random walker may use a shortcut repeatedly,
or it may become stuck in one of the loops (cycles). In the latter case, a
scaling argument shows that long jump lengths are combined with long,
power-law trapping times in cycles, and effectively the overall motion
along the polymer chain is characterized by a \textit{linear time dependence}
of the mean squared displacement~\cite{sokolov1997paradoxal}. Thus,
this last process cannot be considered as a a true LF. A similar situation
is observed in complex networks and in many general graphs where the
random walker can get stuck in small regions of the graph due to
heterogeneities in degree distribution or due to local structural
heterogeneities, such as the existence of network communities, which
will destroy the LF nature. To what extent this happens and what the 
remaining advantage of long-ranged hops is, must be analysed in the actual setting.

In this work we will consider the generalized formulation of a random
walk model with long-range jumps that decay as a function---not necessarily
a power-law---of the shortest path distance between the nodes. We
propose to call this model the \textit{random multi-hopper}
(RMH). We proceed by using the random walk--electrical networks connection 
as discussed in detail by Doyle~\cite{doyle1984random} in order to formulate mathematically the RMH model
and some of its main parameters, namely the hitting and commute times.
We prove that for certain asymptotic values of the
parameters of the model, the average hitting time of the random walker
is the smallest possible. We further study 
some deterministic graphs for which extremal properties of random
walks are known, such as lollipop, barbell and path graphs. Then,
we move to the analysis of random networks, in particular we explore
the Erd\H{o}s-R\'{e}nyi and the Barab\'{a}si-Albert models. We finally
study a few real-world networks representing a variety of complex systems. 
In all cases we compare the RMH with the NRW model and conclude
that the multi-hopper overcomes several of the difficulties that a
normal random walker has to explore a graph. In particular, the multi-hopper
explores graphs with clusters of highly interconnected
nodes and graphs with very skewed degree distributions more efficiently than a normal random walker. 
As these characteristics are omnipresent in real-world networks this makes the random multi-hopper an excellent choice
for transport and search on complex systems.

\section{Preliminaries}

We introduce in this section some definitions and properties
associated with random walks on graphs and set the notation used
throughout the work. A \textit{graph} $G=(V,E)$ is defined by a set
of $n$ nodes (vertices) $V$ and a set of $m$ edges $E=\{(u,v)|u,v\in V\}$
between the nodes. All the graphs considered in this work are finite,
undirected, simple, without self-loops, and connected. A \textit{path}
of length $k$ in $G$ is a sequence of different nodes $u_{1},u_{2},\ldots,u_{k},u_{k+1}$
such that for all $1\leq l\leq k$, $(u_{l},u_{l+1})\in E$. The length
of a shortest path between two nodes $i$ and $j$ constitutes a distance
function, here designated by $d\left(i,j\right)$, which is known
as the \textit{shortest-path distance} between the nodes $i$ and
$j$. Let $A=\left(a_{ij}\right)_{n\times n}$ be the \textit{adjacency
matrix} of the graph where $a_{ij}=1$ if and only if $\left(i,j\right)\in E$
and it is zero otherwise. The \textit{degree} of a node $i$ is
the number of nodes adjacent to it and it is designated here by $k\left(i\right) = \sum_j a_{ij}$.

A \textit{random walk} on a graph is a random sequence of vertices
generated as follows. Given a starting vertex $i$ we select a neighbor
$j$ at random, and move to this neighbor. Then we select a neighbor
$k$ of $j$ uniformly at random, move to it, and so on~\cite{lovasz1993random,aldous2002reversible}.
This sequence of random nodes $v_{t}:t=0,1,\ldots$ is a Markov chain, with probability distribution encoded in the following vector:
\begin{equation}
\mathbf{p}_{t}\left(q\right)=\Pr\left(v_{t}=q\right).
\end{equation}
The \textit{matrix of transition probabilities} is $P=\left(p_{ij}\right)_{i,j\in V}$ is defined via:
\begin{equation}
p_{ij}=\dfrac{a_{ij}}{k(i)}.
\end{equation}

Let $K$ be the diagonal matrix whose entries $K_{ii}=k(i)$.
With this definition we can write the above matrix compactly as $P=K^{-1}A$. 
The vector containing the probability of finding
the walker at a given node of the graph at time $t$ is
\begin{equation}
\mathbf{p}_{t}=\left(P^{T}\right)^{t}\mathbf{p}_{0},
\end{equation}
where $T$ represents the transpose of the matrix and $\mathbf{p}_{0}$
is the initial probability distribution of the random walker.

The following are important characteristics of a random walk on a graph
which are of direct utility in the current work. For a random walk
starting at the node $i$, the expected number of steps before it
reaches the node $j$ is known as the \textit{hitting time} and it
is denoted by $H\left(i,j\right)$. The expected number of steps
in a random walk starting at node $i$, before the random walker visits
the node $j$, and then visits node $i$ again is known as the \textit{commute
time}, and it is denoted by $\kappa\left(i,j\right)$. Both quantities
are related by
\begin{equation}
\kappa\left(i,j\right)=H\left(i,j\right)+H\left(j,i\right).
\end{equation}

\section{Random Multi-Hopper Model}

\subsection{Intuition of the model}

For a normal random walk on a graph, the random walker makes steps
of length one in terms of the number of edges traveled.
After each step she throws again the dice to decide where to move. 
As analogy we may think of a drunkard who does not remember
where her home is. Thus she stops at every junction that she finds
and takes a decision of where to go next in a random way.

Let us now suppose the existence of a random walker who does not necessarily
stop at an adjacent node of her current position. That is, suppose
that the random walker placed at the node $i$ of the graph selects
any node $q$ of the graph to which she wants to move. Let us consider
that the shortest path distance between $i$ and $q$ is $d\left(i,q\right)$.
If $d\left(i,q\right)>1$ the random walker will not stop at any of
the intermediate nodes between $i$ and $q$, but she will go directly
to that node. In our analogy this would correspond to a drunkard who thinks she
remembers where her home is, then she walks a few blocks without stopping
at any junction until she arrives at a given point where she realizes
she is lost. After this she repeats the process again. Therefore the movements
of the drunkard can be represented by a graph, the edges of correspond to shortest
paths in the original graph $G$, and the probability
for the random walker to jump straight from node $i$ to node $q$
is proportional to a certain weight $\omega\left(i,q\right)$, which is
a function of the distance $d\left(i,q\right)$ in $G$. As examples
of decaying functions of the shortest path distance, we mention
$\omega\left(i,q\right)=\exp\left(-l\cdot d\left(i,q\right)\right)$, $\omega\left(i,q\right)=\left(d\left(i,q\right)\right)^{-s}$
and $\omega\left(i,q\right)=z^{-d\left(i,q\right)}$, for $l>0$, $s>0$,
and $z>1$, respectively. Hereafter we will consider the first two
for the analysis, called respectively the Laplace transform case and
the Mellin transform case, by analogy with these transforms. 
The specific 
details in which we implement these transforms are given in the next section, 
where we used a slight variation of the Laplace transform.

As an example let us consider a one-dimensional linear chain of 5 nodes labeled as 1---2---3---4---5, 
and a weight function
$\omega\left(i,q\right)=\exp\left(-0.5\cdot d\left(i,q\right)\right)$.
If the drunkard is placed at the node 1 she has the following probabilities
of having a walk of length 1, 2, 3 or 4: 0.46, 0.28, 0.18 and 0.10,
respectively. That is, the probability that she stops at the
nearest node from her current position is still higher than that for the
rest of nodes, but the last ones are not zero like in the classical
random walk. If $l\to\infty$, the probabilities approach those of 
the classical random walker. For instance if $\omega\left(i,q\right)=\exp\left(-5\cdot d\left(i,q\right)\right)$
the probabilities of having a walk of length 1, 2, 3 or 4 are: 0.9933,
0.0067, 0.000 and 0.000, respectively. For $l>10$ we effectively
recover the classical random walk model, and have numerical probabilities of 1, 0,
0, and 0 for the walks of length 1, 2 , 3, and 4, respectively. This
shows how the current model is a generalization of the
classical random walk model.

\subsection{Mathematical formulation}

Consider a connected graph $G=\left(V,E\right)$ with $n$ nodes. 
Let $d_\textrm{max}$ be the graph diameter, i.e., the maximum shortest path
distance in the graph. Let us now define the $d$-path adjacency matrix
($d\leq d_\textrm{max}$), denoted by $A_{d}$, as the symmetric $n\times n$ matrix, with entries:
\begin{equation}
A_{d}\left(i,j\right)=\left\{ \begin{array}{c}
1\\
0
\end{array}\right.\begin{array}{c}
if\ d_{ij}=d,\\
otherwise,
\end{array}
\end{equation}
where $d_{ij}$ is the shortest path distance, i.e., the number of
edges in the shortest path connecting the nodes $i$ and $j$.
The $d$-path degree of the node $i$ is given by~\cite{estrada2012path,riascos2012long}
\begin{equation}
    k_{d}\left(i\right)=\left(A_{d}\mathbf{1}\right)_{i}
\end{equation}
where $\mathbf{1}$ is an all-ones column vector. 

Let us now consider the following transformed $k$-path adjacency
matrices~\cite{estrada2012path}:
\begin{equation}\label{eq:generalized operator}
\hat{A}^{\tau}=
\begin{cases}
\sum_{d=1}^{d_\textrm{max}}d^{-s}A_{d}
&\quad\mbox{for $\tau=\textnormal{Mel}$,} \\
A_1+\sum_{d=2}^{d_\textrm{max}}\exp(-l\cdot d)A_{d}
&\quad\mbox{for $\tau=\textnormal{Lapl}$,}
\end{cases}
\end{equation}
with $s>0$ and $l>0$, which are the Mellin and Laplace transforms, respectively.

Let us define the generalized degree of a given node as
\begin{equation}
\hat{k}^{\tau}\left(i\right)=\left(\hat{A}^{\tau}\mathbf{1}\right)_{i}.\label{eq:generalized degree}
\end{equation}
Now we define the probability that a particle staying at node $i$
hops to the node $j$ as
\begin{equation}
P^{\tau}\left(i,j\right)=\frac{\hat{A}^{\tau}\left(i,j\right)}{\hat{k}^{\tau}\left(i\right)}.
\end{equation}
Notice that if do not consider any long-range interaction, then $\hat{A}^{\tau}=A$
and $P^{\tau}\left(i,j\right)= P$. That is, we recover the classical random walk probability.

Let us denote by $\hat{K}^{\tau}$ the diagonal matrix with $\hat{K}^{\tau}\left(i,i\right)=\hat{k}^{\tau}\left(i\right)$
and let us define the matrix $\hat{P}^{\tau}=(\hat{K}^{\tau})^{-1}\hat{A}^{\tau}$.
Then, the evolution equation ruling the states of the walker at a
given time step is given by
\begin{equation}
\mathbf{p}_{t+1}=\left(\hat{P}^{\tau}\right)^{T}\mathbf{p}_{t}.
\end{equation}

\subsection{$k$-path Laplacians, Hitting and Commute Times}

In a similar way as the Laplacian matrix is introduced for graph we
define the Laplacian matrix corresponding to Eq.~(\ref{eq:generalized operator}).
That is,
\begin{equation}
\hat{L}^{\tau}=\hat{K}^{\tau}-\hat{A}^{\tau},\label{eq:k-path Laplacian}
\end{equation}
where $K^{\tau}$ is the diagonal matrix of generalized degree $\hat{k}^{\tau}\left(i\right)$
defined in Eq.~(\ref{eq:generalized degree}) and $A^{\tau}$ is the generalized
adjacency matrix defined in Eq.~(\ref{eq:generalized operator}). 
This is the Laplacian of the graph $G^{\tau}$, the (weighted) adjacency
matrix of which is $A^{\tau}$. As a result, this generalized Laplacian
$\hat{L}^{\tau}$ is positive semi-definite~\cite{estrada2012path}.

The graph $G^{\tau}$ can be seen as a network of resistances, with
the entry of $A^{\tau}$ representing the conductance (inverse resistance)
of the connection between two nodes. By assuming that an electric
current is flowing through the network $G^\tau$ by entering at node
$i$ and leaving at node $j$ we can calculate the effective resistance
between these two nodes as follows:
\begin{equation}
\hat{\varOmega}^{\tau}\left(i,j\right)=\hat{\mathcal{L}}^{\tau}\left(i,i\right)+\hat{\mathcal{L}}^{\tau}\left(j,j\right)-2\hat{\mathcal{L}}^{\tau}\left(i,j\right),
\end{equation}
where $\hat{\mathcal{L}}^{\tau}$ is the Moore-Penrose pseudo-inverse
of the generalized Laplacian matrix. It is well known that the analogous
of this effective resistance for the simple graph is a distance between
the corresponding pair of nodes~\cite{klein1993resistance}. It is straightforward to show that
this is also the case here and we will call $\hat{\varOmega}^{\tau}\left(i,j\right)$
the \textit{generalized effective resistance} between the nodes $i$ and $j$ in a graph.

The sum of all resistance distances in a graph is know as the Kirchhoff
index of the graph~\cite{klein1993resistance}. In the context of the multi-hopper model it can
be defined in a similar way as
\begin{equation}
\hat{\varOmega}^{\tau}_\mathrm{tot}=\sum_{i<j}\hat{\varOmega}^{\tau}\left(i,j\right)=\dfrac{1}{2}\mathbf{1}^{T}\hat{\varOmega}^{\tau}\mathbf{1}.
\end{equation}
Then, an extension of a result obtained by Nash-Williams~\cite{nash1959random}
and by Chandra \textit{et al}.~\cite{chandra1996electrical} allows us to calculate
the commute and hitting times based on the generalized resistance
distance. That is, the commute time between the corresponding nodes
is given by
\begin{equation}
\hat{\kappa}^{\tau}\left(i,j\right)=vol\left(G^{\tau}\right)\hat{\varOmega}^{\tau}\left(i,j\right),
\end{equation}
where $vol\left(G^{\tau}\right)$ is the sum of all the weights of
the edges of $G^{\tau}$ (see for instance, Ref.~\cite{ghosh2008minimizing}). 
Using the \textit{scaled generalized Fundamental Matrix} (SGFM) (details about the definition and notation for this matrix
can be found in Refs.~\cite{grinstead2012introduction,boley2011commute}) we can express
the hitting and commute times in matrix form as
\begin{equation}
\hat{H}^{\tau}=\mathbf{1}\left[diag\left(\tilde{Z^{\tau}}^{-1}\right)\right]^{T}-\tilde{Z^{\tau}},
\end{equation}
\begin{equation}
\hat{\kappa}^{\tau}=\mathbf{1}\left[diag\left(\tilde{Z^{\tau}}^{-1}\right)\right]^{T}+\left[diag\left(\tilde{Z^{\tau}}^{-1}\right)\right]\mathbf{1}^{T}-\tilde{Z^{\tau}}-\tilde{Z^{\tau}}^{T},
\end{equation}
where $\tilde{Z^{\tau}}$ is just the SGFM for the graph $G^{\tau}$. 
The expected commute time averaged over all pairs of nodes can be
easily obtained from the multi-hopper Kirchhoff index as
\begin{equation}
\left\langle \hat{\kappa}^{\tau}\right\rangle =\dfrac{4\left(\mathbf{1}^{T}\mathbf{w^{\tau}}\right)}{n\left(n-1\right)}\hat{\varOmega}^{\tau}_\mathrm{tot},
\end{equation}
where $\mathbf{w^{\tau}}$ is the vector containing the weight of
each edge in the graph $G^{\tau}$.
In a similar way we can obtain the expected hitting time averaged
over all pairs of nodes
\begin{equation}\label{eq3.14}
\left\langle \hat{H}^{\tau}\right\rangle =\dfrac{2\left(\mathbf{1}^{T}\mathbf{w^{\tau}}\right)}{n\left(n-1\right)}\hat{\varOmega}^{\tau}_\mathrm{tot}.
\end{equation}

In order to understand the mechanism behind the efficiency of the
multi-hopper random walker to explore networks, it is important to
relate the structure of the network with dynamical quantities. In
the following part we study the stationary probability distribution
and the mean-first return time and its relation with the distances
in the network. \\[2mm] The stationary probability distribution vector
$\boldsymbol{\pi}^{\tau}$ is obtained as follows:
\begin{equation}
\boldsymbol{\pi}^{\tau}=\frac{\hat{P}^{\tau}\mathbf{1}}{\mathbf{1}^{T}\hat{P}^{\tau}\mathbf{1}}.
\end{equation}
Now, having into account the definition of the matrix $\hat{P}^{\tau}$,
we obtain for the elements $\pi^{\tau}(i)$ with $i=1,2,\ldots,n$
of the stationary probability distribution: 
\begin{equation}
{\pi}^{\tau}(i)=\frac{k^{\tau}(i)}{\sum_{j=1}^{n}k^{\tau}(j)},\label{StatDistf}
\end{equation}
In this way, we obtain for the Mellin transformation with parameter
$s$: 
\begin{equation}
k^{\tau=\textnormal{Mel}}(i)=k\left(i\right)+k_{2}\left(i\right)\frac{1}{2^{s}}+k_{3}\left(i\right)\frac{1}{3^{s}}+\ldots+k_{d_\textrm{max}}\left(i\right)\frac{1}{d_\textrm{max}^{s}},\label{LRDegreeM}
\end{equation}
and the Laplace transformation with parameter $l$: 
\begin{equation}
k^{\tau=\textnormal{Lapl}}(i)=k\left(i\right)+k_{2}\left(i\right)\frac{1}{e^{2l}}+k_{3}\left(i\right)\frac{1}{e^{3l}}+\ldots+k_{d_\textrm{max}}\left(i\right)\frac{1}{e^{l d_\textrm{max}}}.\label{LRDegreeE}
\end{equation}
The stationary probability distribution in Eq.~(\ref{StatDistf})
determines the probability to find the random walker at the node $i$
in the limit $t$ large. The expressions~(\ref{LRDegreeM}) and~(\ref{LRDegreeE})
allow to identify how the structure of the networks and the long-range
strategy controlled by the parameters $s$ or $l$ combine in order
to change the stationary probability distribution. In the limit of
$s,l\to\infty$ the long-range contribution is null and the result
$\pi(i)=k_{i}/\sum_{l=1}^{n}k_{l}$ for the normal random
walker is recovered. On the other hand, when $s,l\to0$, the dynamics
includes, in the same proportion, contributions of first-, second-,
third-,..., and $d_\textrm{max}$-nearest neighbors. In this limit case the
stationary probability distribution is the same for all the nodes
and $\pi(i)=1/n$.

\section{On the hitting time in the random multi-hopper model}

The most important result of this work is related to the average hitting
time of the random multi-hopper walk when the parameters of the corresponding
transforms tends to zero.
\begin{lemma}
Let us consider the transformed $k$-path adjacency matrices:
\begin{equation}
\hat{A}^{\tau}=\sum_{d=1}^{d_\textrm{max}}c_{d}^{\tau}A_{d},\label{eq:generalized operator-1}
\end{equation}
with $c_{d}^{\tau=\textnormal{Mel}}=d^{-s}$ for $s>0$ and $c_{d}^{\tau=\textnormal{Lapl}}=\exp\left(-l\cdot d\right)$
for $l>0$. Then, when $s\rightarrow0$ or $l\rightarrow0$ the average
hitting time $\langle \hat{H}^{\tau}\rangle\rightarrow\left(n-1\right)$,
independently of the topology of the graph, which is the minimum for
any graph of $n$ nodes.
\end{lemma}
\begin{proof}
First, let $G_{n}$ be any connected, simple graphs with $n$ nodes.
Then, $\left\langle \hat{H}\left(G_{n}\right)\right\rangle \leq\left\langle \hat{H}\left(K_{n}\right)\right\rangle =n-1$,
with equality if and only if $G_{n}=K_{n}$, where $K_{n}$ is the
complete graph with $n$ nodes (see Ref.~\cite{palacios2001resistance}). 
Then, when $s\rightarrow0$ or $l\rightarrow0$
we have that $\hat{A}^{\textnormal{Mell}}\left(i,j\right)\rightarrow1$
and $\hat{A}^{\textnormal{Lapl}}\left(i,j\right)\rightarrow1$, respectively.
This means that $\hat{A}^{\tau}\left(i,j\right)=1$ $\forall i\neq j$
and $\hat{A}^{\tau}\left(i,j\right)=0$ $\forall i=j$. In other words,
$\hat{A}^{\tau}=\mathbf{1}\mathbf{1}^{T}-I$, which is the adjacency
matrix of the complete graph $K_{n}$. In closing, when $s\rightarrow0$
or $l\rightarrow0$, $\hat{A}^{\tau}\rightarrow A\left(K_{n}\right)$.
As it has been previously proved $\left\langle H\left(K_{n}\right)\right\rangle =n-1$~\cite{palacios2001resistance},
which proves the result.
\end{proof}

\subsection{Case study: hitting times of the multi-hopper on all 8-node graphs}
We study the average hitting time of all $11,117$ connected
graphs with 8 nodes. The average hitting time has mean $10.036\pm1.932$
for all the graphs with $n=8$, with a maximum value of $21.071$.
These values converge quickly to $n-1$ as soon as $s,l\rightarrow0$.
For instance for $s=0.5$ the mean of the average hitting time is
already $7.062\pm0.024$, and this value drops up to $7.00253\pm0.00096$
for $s=0.1$. The situation is very similar for $l\rightarrow0$,
and the mean of the average hitting time is $7.0037\pm0.0045$ for
$l=0.1$ and $7.000072\pm4.67\cdot10^{-5}$ for $l=0.01$. In all
cases the minimum value of the average hitting time is obtained for
the complete graph $K_{8}$. 

However, this is where the similarities
between the classical RW and the multi-hopper with the Mellin and
the Laplace transforms end. For instance, the graph with the maximum
value of $\left\langle H\right\rangle $ for the NRW is the lollipop
graph $L\left(8,4\right)$ illustrated in the left panel of Fig.~\ref{Maximal graphs}. 
However, for $\langle \hat{H}^{\textnormal{Mell}}\rangle $
with the values of $s$ studied here the maximum is reached for the
graph formed by a diamond graph with two paths connected to opposite
nodes as illustrated in the right panel of Fig.~\ref{Maximal graphs}. For
$\langle \hat{H}^{\textnormal{Lapl}}\rangle $ with the
values of $l$ studied here, the graph displaying the maximum is the
path graph $P_{8}$. We unfortunately do not have an analytic or otherwise
explanation for this phenomenon.
\begin{figure}
\begin{center}
\includegraphics[width=0.3\textwidth]{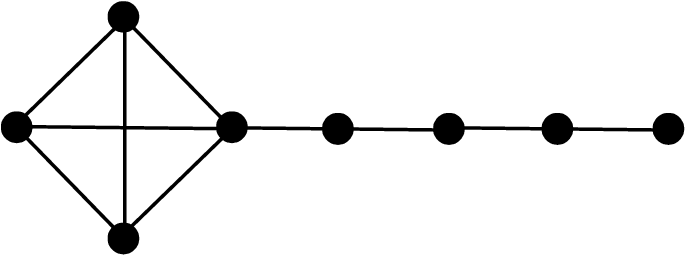}
\hspace{0.1\textwidth}
\includegraphics[width=0.3\textwidth]{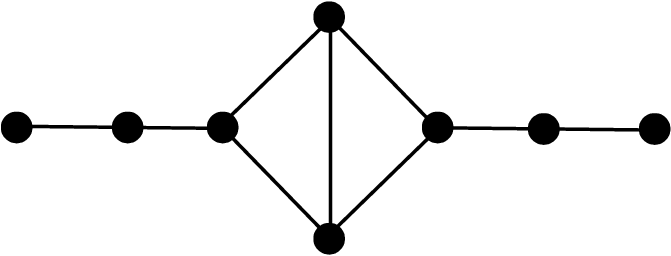}
\end{center}
\caption{Illustration of the lollipop graph $L\left(8,4\right)$ (left) and the graph
displaying the maximum value of $\langle \hat{H}^{\textnormal{Mell}}\rangle $ (right). }
\label{Maximal graphs}
\end{figure}

The lollipop graphs appear in many extremal properties related to
random walks on graphs. In 1990 Brightwell and Winkler~\cite{brightwell1990maximum}
proved that the hitting time between a pair of nodes $i$ and $j$
in a graph is maximum for the lollipop graph $L\left(n,\left\lfloor \tfrac{2n}{3}\right\rfloor \right)$
consisting of a clique of $\lfloor \frac{2n}{3}\rfloor $
nodes including $i$ to which a path on the remaining nodes, ending
in $j$, is attached. The same graph was found by Jonasson as the
one containing the pair of nodes maximizing the commute time among
all graphs~\cite{jonasson2000lollipop}.

Then, we investigate the average hitting time of the lollipops $L(n,\lfloor \tfrac{n}{2}\rfloor )$---the
one having the largest value of $\langle H\rangle $ for
the NRW among the graphs with $n=8$---and $L\left(n,\left\lfloor \tfrac{2n}{3}\right\rfloor \right)$,
as well as the graphs with the structure displayed in the right panel of Fig.~\ref{Maximal graphs}
and the paths $P_{n}$ for $8\leq n\leq1000$. Our investigation
produces the following results: (i) among the four types of graphs studied
the lollipop $L\left(n,\left\lfloor \tfrac{n}{2}\right\rfloor \right)$
has always the largest value of $\left\langle H\right\rangle $; (ii)
for $n\geq10$, the lollipop $L\left(n,\left\lfloor \tfrac{2n}{3}\right\rfloor \right)$
has always the largest value of $\langle \hat{H}^{\textnormal{Mell}}\rangle $
for $s=0.5$---the graph of the type illustrated in the right panel of Fig.~\ref{Maximal graphs}
is maximum for $n=8,9$; (iii) for $8\leq n\leq29$ the path
graph $P_{n}$ has the largest value of $\langle \hat{H}^{\textnormal{Lapl}}\rangle $
for $l=0.1$. For $n\geq30$ the lollipop $L\left(n,\left\lfloor \tfrac{n}{2}\right\rfloor \right)$
has the largest value of $\langle \hat{H}^{\textnormal{Lapl}}\rangle $
for $l=0.1$ among the four types of graphs investigated. 

On the basis of the previous results one may conjecture that the
graphs with the maximum value of the average hitting time correspond
to certain type of lollipop graph. This is intuitively justified
by the fact that in a lollipop the random walker spend a lot of time
hopping among the nodes of the clique and it has a little chance of
going through the path. However, on the same intuitively basis the
barbell graphs are also good candidates for having the maximum average
hitting time among graphs with $n$ nodes (see for instance Ref.~\cite{aldous2002reversible}).
A barbell graph $B\left(n,k_{1},k_{2}\right)$ is a graph with $n$
nodes and two cliques of sizes $k_{1}$ and $k_{2}$ connected by
a path on the remaining nodes. The problem that arises is to determine which
of the many lollipop or barbell graphs has the largest
average hitting time. This problem is out of the scope of the current
work, but constitutes an important open problem. We remark
that even for the case of the NRW with its many years of investigation,
it is still unknown which is the graph with the maximum average hitting
time. In case of the multi-hopper the problem is even more complicated as the
maximum value of the average hitting time may
depend on the value of the parameters $s$ and $l$ of the transform
used. We thus formulate the following open problem.
\begin{problem}
Determine the graph(s) with the maximum value of $\langle \hat{H}^{\tau}\rangle $
for different transforms of the multi-hopper random walk.
\end{problem}

\section{Deterministic graphs}

In this section we study some of the properties of the multi-hopper
model for some classes of graphs which have deterministic structure.

\subsection{Lollipop and barbell graphs}

The first classes of graphs that we study here are the so-called lollipop
and barbell graphs. These graphs appear in many extremal properties
related to random walks on graphs as we have discussed in the previous
section. Here we consider lollipop graphs $L\left(n,\left\lfloor \tfrac{n}{2}\right\rfloor \right)$
and $L\left(n,\left\lfloor 2\tfrac{n}{3}\right\rfloor \right),$which
have appeared already in the previous section and the symmetric barbell
graphs $B\left(n,\left\lfloor \tfrac{n}{3}\right\rfloor ,\left\lfloor \tfrac{n}{3}\right\rfloor \right)$.
We want to remark that these graphs are not necessarily the extremal
ones for the hitting time as discussed in the previous section but
they can be considered as representative of their classes. 

We observe that the three graphs display $\left\langle H\right\rangle \approx an^{3}$
for the normal random walk. The coefficients $a$ obtained by using
nonlinear fitting of the hitting times with $n$ are: $a\approx0.01387$
for $L\left(n,\left\lfloor \tfrac{n}{3}\right\rfloor \right)$, $a\approx0.0179$
for $B\left(n,\left\lfloor \tfrac{n}{3}\right\rfloor ,\left\lfloor \tfrac{n}{3}\right\rfloor \right)$
and $a\approx0.01928$ for $L\left(n,\left\lfloor \tfrac{2n}{3}\right\rfloor \right)$. 

We then study the variation of the parameters $s$ and $l$ in the
Mellin and Laplace transforms of the multi-hopper model for the three
graphs previously studied. In Fig.~\ref{Lollipops_Barbell} we
illustrate the results of these calculations. The important things
to remark in this point are the obvious differences between the use
of the Mellin and Laplace transforms in the multi-hopper model. First,
it is easily observed that the Mellin transform produces a faster
decay of the average hitting time than the Laplace one for the three
graphs. For instance the 50\% reduction in the average hitting time
of the three graphs occurs for values of $3.5<s<4.0$ for the Mellin
transform, but it happens for $1.5<l<2.0$ for the Laplace. This implies
that the Laplace transform converges to the average hitting time of $n-1$
at much smaller values than the Mellin transform. The second important
difference is observed in the insets of Fig.~\ref{Lollipops_Barbell}.
For the Laplace-transformed multi-hopper, the lollipop graph $L\left(n,\left\lfloor \tfrac{n}{2}\right\rfloor \right)$
always has the largest value of the average hitting time at any value
of $l$ among the three graphs studied. However, for the Mellin-transformed
case, the lollipop graph $L\left(n,\left\lfloor \tfrac{n}{2}\right\rfloor \right)$
has the largest hitting time for large values of $s$, but for $s\lesssim2.1$
the lollipop graph $L\left(n,\left\lfloor 2\tfrac{n}{3}\right\rfloor \right)$
is the one with the largest value of the average hitting time (see
the crossing in the inset of Fig.~\ref{Lollipops_Barbell}). This
confirms the complexity of the analysis of the extremal graphs for
the multi-hopper model as we have hinted in the previous section.
\begin{figure}
\includegraphics[width=0.45\textwidth]{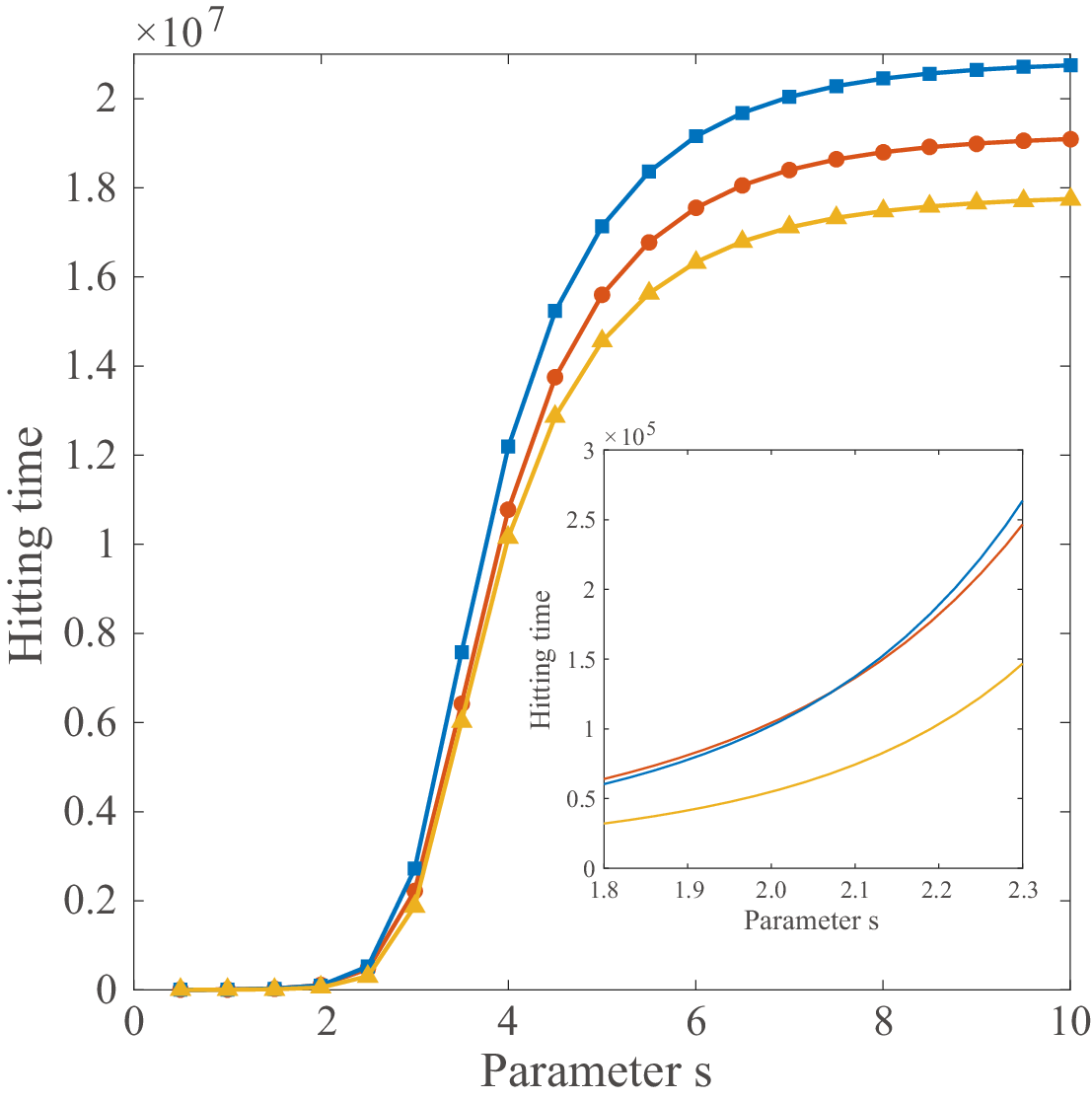}
\hfill
\includegraphics[width=0.45\textwidth]{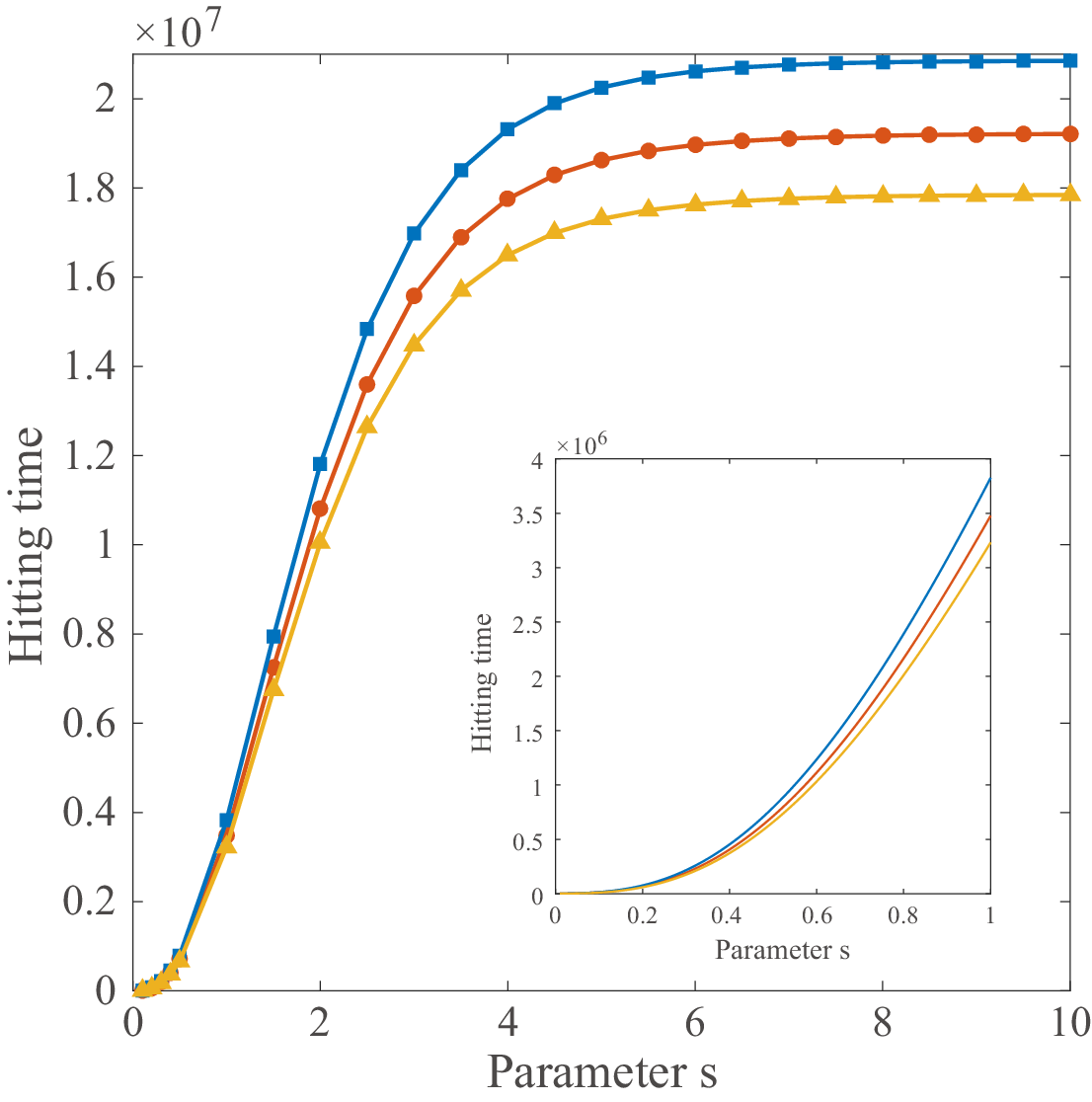}
\caption{Hitting time as a function of the parameter $s$ for the Mellin (left)
and the Laplace (right) transforms in lollipops $L\left(n,\left\lfloor \tfrac{n}{2}\right\rfloor \right)$
(blue squares), $L\left(n,\left\lfloor 2\tfrac{n}{3}\right\rfloor \right)$
(red circles) and barbell $B\left(n,\left\lfloor \tfrac{n}{3}\right\rfloor ,\left\lfloor \tfrac{n}{3}\right\rfloor \right)$
(yellow triangles) graphs for $n=999$. In the inset panels we zoom
the plot for the region $1.8\leq s\leq2.3$ and $0.01\leq l\leq1$,
respectively. }
\label{Lollipops_Barbell}
\end{figure}

We now concentrate on the variation of the average hitting time with
the number of nodes in the lollipop $L\left(n,\left\lfloor \tfrac{2n}{3}\right\rfloor \right)$
for different values of the parameters $s$ and $l$ (see Fig.~\ref{Lollipop_Size}).
For the Mellin transformed multi-hopper model with a fixed value of
the exponent $s$, the average hitting time can always be fit well by a power-law
of the number of nodes: $\langle \hat{H}^{\textnormal{Mel}}\left(s\right)\rangle \approx an^{\gamma}$,
where $\gamma\rightarrow3$ when $s\rightarrow\infty$ and $\gamma\rightarrow1$
when $s\rightarrow0$. For instance, $\gamma=2.831$ for $s=3;$ $\gamma=2.129$
for $s=2;$ $\gamma=1.772$ for $s=1;$ $\gamma=1.291$ for $s=0.5;$
$\gamma=1.011$ for $s=0.1$. The situation is quite similar for the
Laplace transformed multi-hopper model. This observation indicates
that for small values of the parameters $s$ and $l$ the average
hitting time changes linearly with the number of nodes. This important
observation is repeated for every family of graphs as we will see
in further sections of this work.
\begin{figure}
\includegraphics[width=0.45\textwidth]{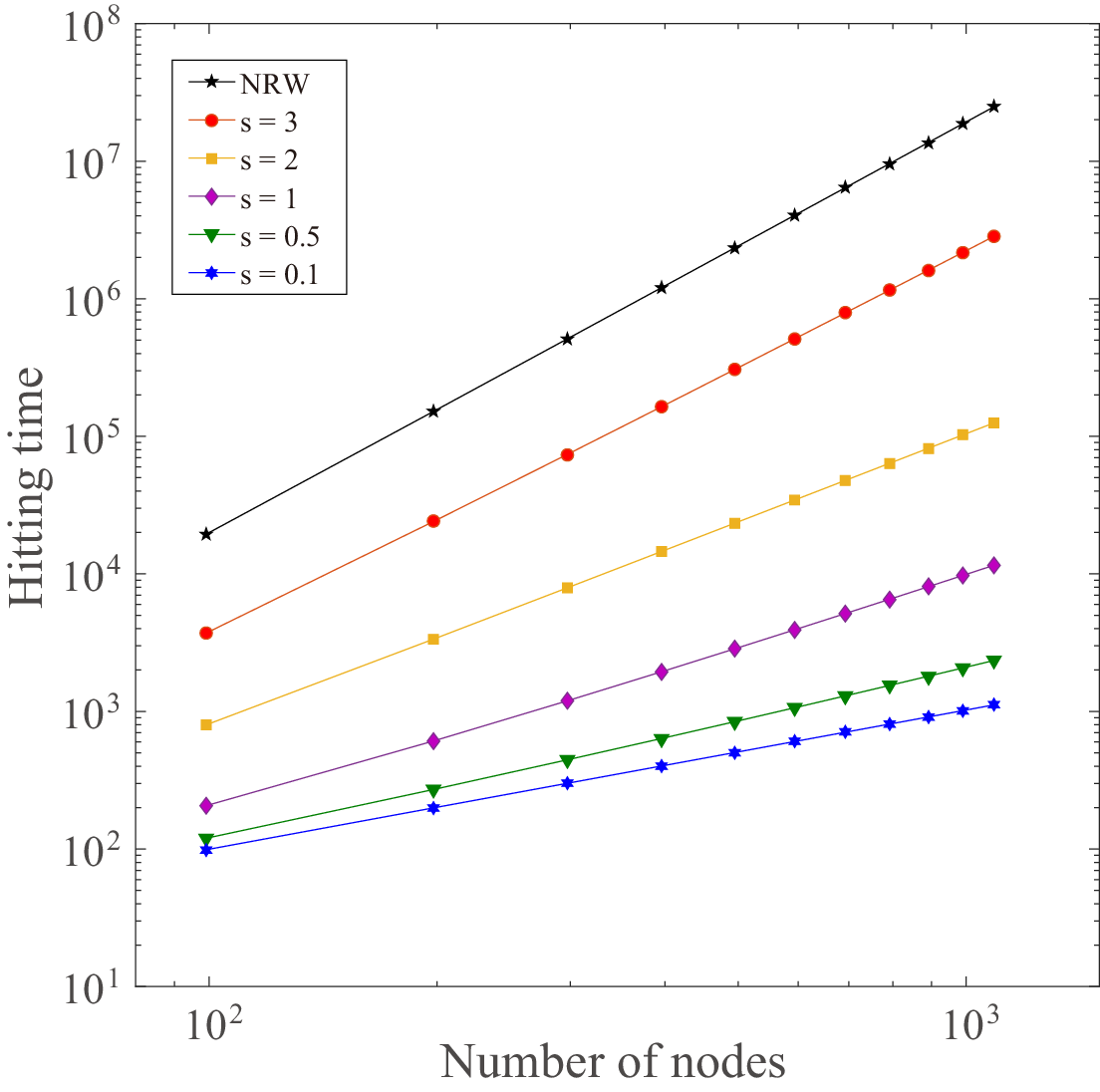}
\hfill
\includegraphics[width=0.45\textwidth]{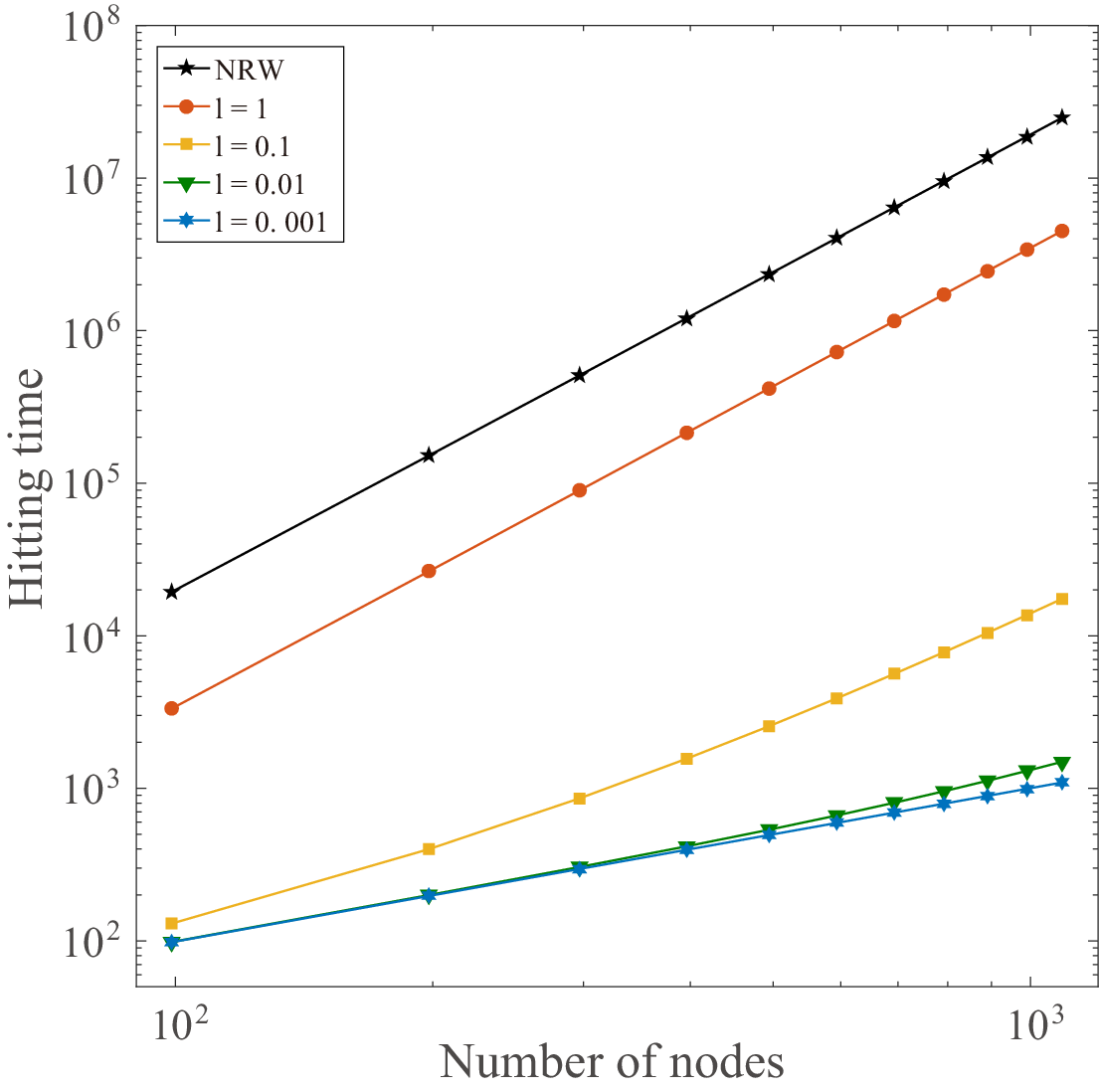}
\caption{Average hitting times for the random multi-hopper model with Mellin
(left) and Laplace (right) transforms in the lollipop graph $L\left(n,\left\lfloor \tfrac{2n}{3}\right\rfloor \right)$
as a function of the number of nodes $n$ in the graph. }
\label{Lollipop_Size}
\end{figure}

We then study the variation of the average hitting time with the number
of nodes for the lollipop $L\left(n,\left\lfloor \tfrac{2n}{3}\right\rfloor \right)$
for $0.001\leq s\leq0.05$ and obtain a linear dependence of the
type: $\langle \hat{H}^{\textnormal{Mell}}\left(G\right),s\rangle \approx\alpha n+\beta$. 

Using these linear fits we can estimate the critical number of nodes
$n_{c}$ below which $\left(n-1\right)\leq\langle \hat{H}^{\textnormal{Mell}}\left(L\left(n,\left\lfloor \tfrac{2n}{3}\right\rfloor \right)\right),s\rangle \leq n$
for a given value of the parameter $s$. Clearly, $n_{c}\leq\beta/(1-\alpha)$.
However, we can simplify this expression by observing that $\beta<-1$
and that $\alpha\approx1+2.751s^{2}$. Then, 
\begin{equation}
n_{c}\leq\dfrac{1}{2.751s^{2}},
\quad
s\leq0.05.
\end{equation}

The values of the critical number of nodes 
range from $145$ for $s=0.05$ to $363,504$ for $s=0.01.$ This
means, for instance, that any lollipop $L\left(n,\left\lfloor \tfrac{2n}{3}\right\rfloor \right)$
having less than 58,160 nodes will display average hitting time bounded
between $n-1$ and $n$ in the random multi-hopper model with a Mellin
transform and parameter $s\leq0.0025$. The previous inequality can
also be used in the other way around, namely in order to estimate
what is the value of $s$ that should be used such that a lollipop
$L\left(n,\left\lfloor \tfrac{2n}{3}\right\rfloor \right)$ has average
hitting time bounded as $\left(n-1\right)\leq\langle \hat{H}^{\textnormal{Mell}}\left(L\left(n,\left\lfloor \tfrac{2n}{3}\right\rfloor \right)\right),s\rangle \leq n$.
For instance, if we would like to know the value of $s$ for which
any graph with less than 100,000 nodes has hitting time below $n-1$,
we use 
\begin{equation}
s\geq\dfrac{1}{\sqrt{2.751n_{c}}},
\quad
s\leq0.05,\label{eq:s_bound}
\end{equation}
and obtain $s\approx0.0019$. We venture out here and make some
extrapolations to roughly estimate the value of $s$ for which any
lollipop $L\left(n,\left\lfloor \tfrac{2n}{3}\right\rfloor \right)$
with less than 1 million nodes has hitting time bounded as before.
This estimation gives a value of $s\lesssim0.0006$. 

The importance of the previous investigation is the following. Currently
we do not know what are the graphs with the largest value of the average
hitting time among all the graphs with $n$ nodes. However, we have
evidence that it should be either a lollipop
or a barbell graph. For these graphs the average hitting time is of
the order $n^{3}$ for the NRW. Then, we can use the previous values
obtained for the lollipop $L\left(n,\left\lfloor \tfrac{2n}{3}\right\rfloor \right)$
as rough indications of the worse case scenarios that can be expected
for any graph. In other words, if we consider a graph of any structure
having 1 million nodes we should expect that its average hitting time
is bounded below $n$ for $s\lesssim 0.0006$. We will see that for
the case of real-world networks, these values of $s$ are orders of
magnitude over-estimated in relation
to this upper bound. A first flavor of these differences is obtained
by the analysis of random graphs in the next section of this work.

\subsubsection{Time evolution}

In this section we are not interested in a detailed description of
the time evolution of the random walker or the multi-hopper in the
lollipop or barbell graphs. We rather will make a comparison between
the evolution of them at different times in such a way that we remark
the main difference between the two models. Consequently we consider
a lollipop $L\left(n,\left\lfloor \tfrac{2n}{3}\right\rfloor \right)$
and a barbell $B\left(n,\left\lfloor \tfrac{n}{3}\right\rfloor ,\left\lfloor \tfrac{n}{3}\right\rfloor \right)$
graph, both with $n=999$ nodes. In both cases we place the random
walker at a node of one of the two cliques. This node is selected
not to be the one attached to the path. Let any node of a clique in
the lollipop (respectively, a clique in the barbell) which is not
the one connected to the path be designated as the node $i$. Let
the endpoint of the path be named $j$. We remark that the node $j$
does not belong to the clique. Let the node connecting the clique
to which $i$ belongs and the path be designated as $k$. Then, we
have placed the random walker and the multi-hopper at the node $i$
of the lollipop and the barbell and explore the probability at each
node after different times using
\begin{equation}
\mathbf{p}_{t}=\left(\tilde{P}^{T}\right)^{t}\mathbf{p}_{0},
\end{equation}
where $\tilde{P}^{T}$ is the transpose of $\hat{P}^{\textrm{Mel}}.$

As can be seen in Fig.~\ref{Barbell_time evolution-1} the
classical random walker spends most of its time in the clique of the
graphs, taking on average $\left\lfloor \tfrac{2n}{3}\right\rfloor -1$
steps to visit the node $k$ in the lollipop and $\left\lfloor \tfrac{n}{3}\right\rfloor -1$
to visit it in the barbell. Once the walker visits the node $k$ she
can walk to the node $j$ only with probability $1/\left\lfloor \tfrac{2n}{3}\right\rfloor $
in the lollipop and $1/\left\lfloor \tfrac{n}{3}\right\rfloor $ in
the barbell graph. Then it can be seen in  Fig.~\ref{Barbell_time evolution-1}
that for times as large as $t=10^{6}$ the random walker is still
stacked in the clique of the lollipop graph. In the case of the barbell
when $t=10^{3}$ the walker has visited only the nodes of the clique
in which she started and when $t=10^{6}$ she starts to visit the
nodes of the other clique.

\begin{figure}
\includegraphics[width=0.32\textwidth]{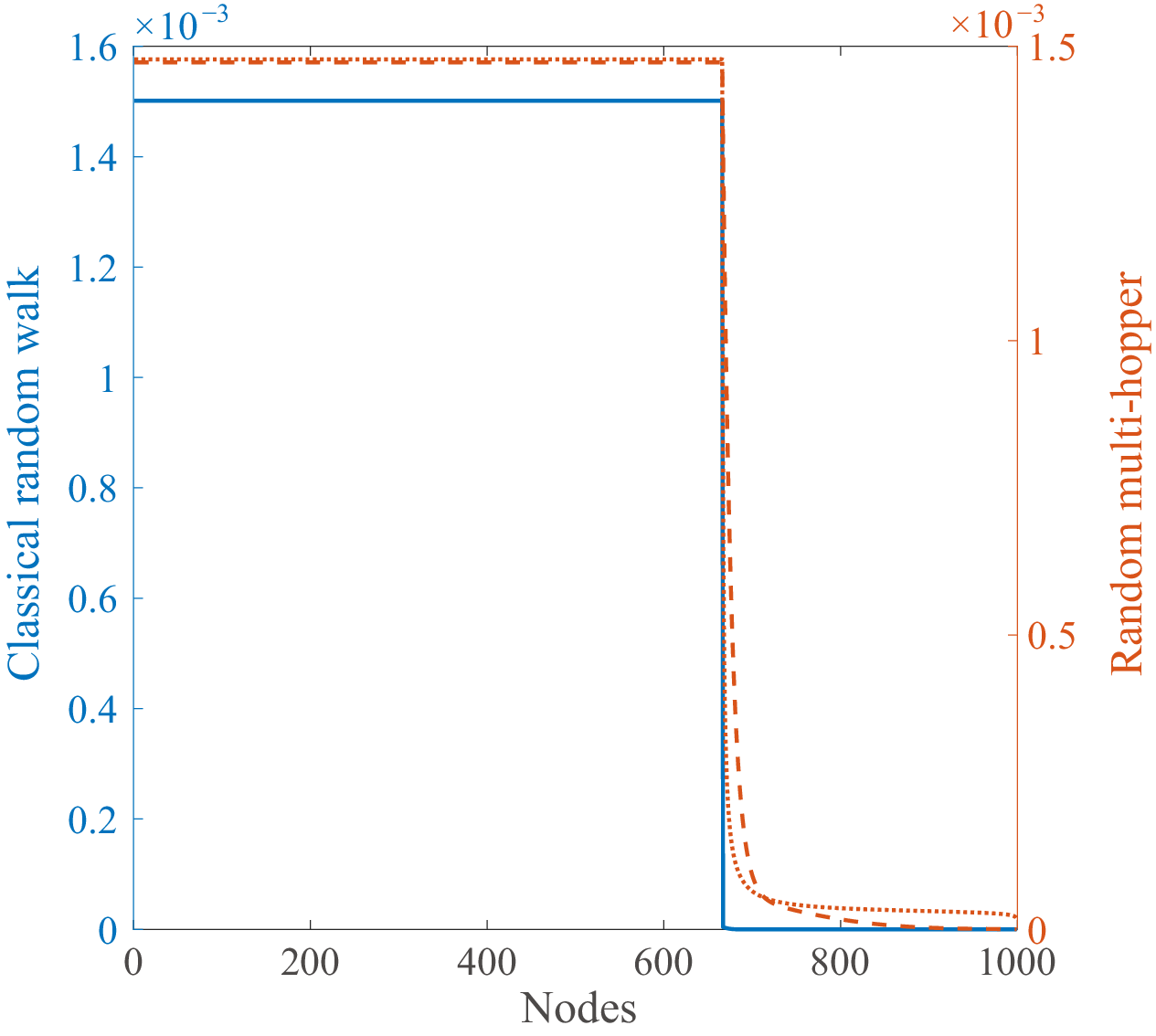}
\hfill
\includegraphics[width=0.32\textwidth]{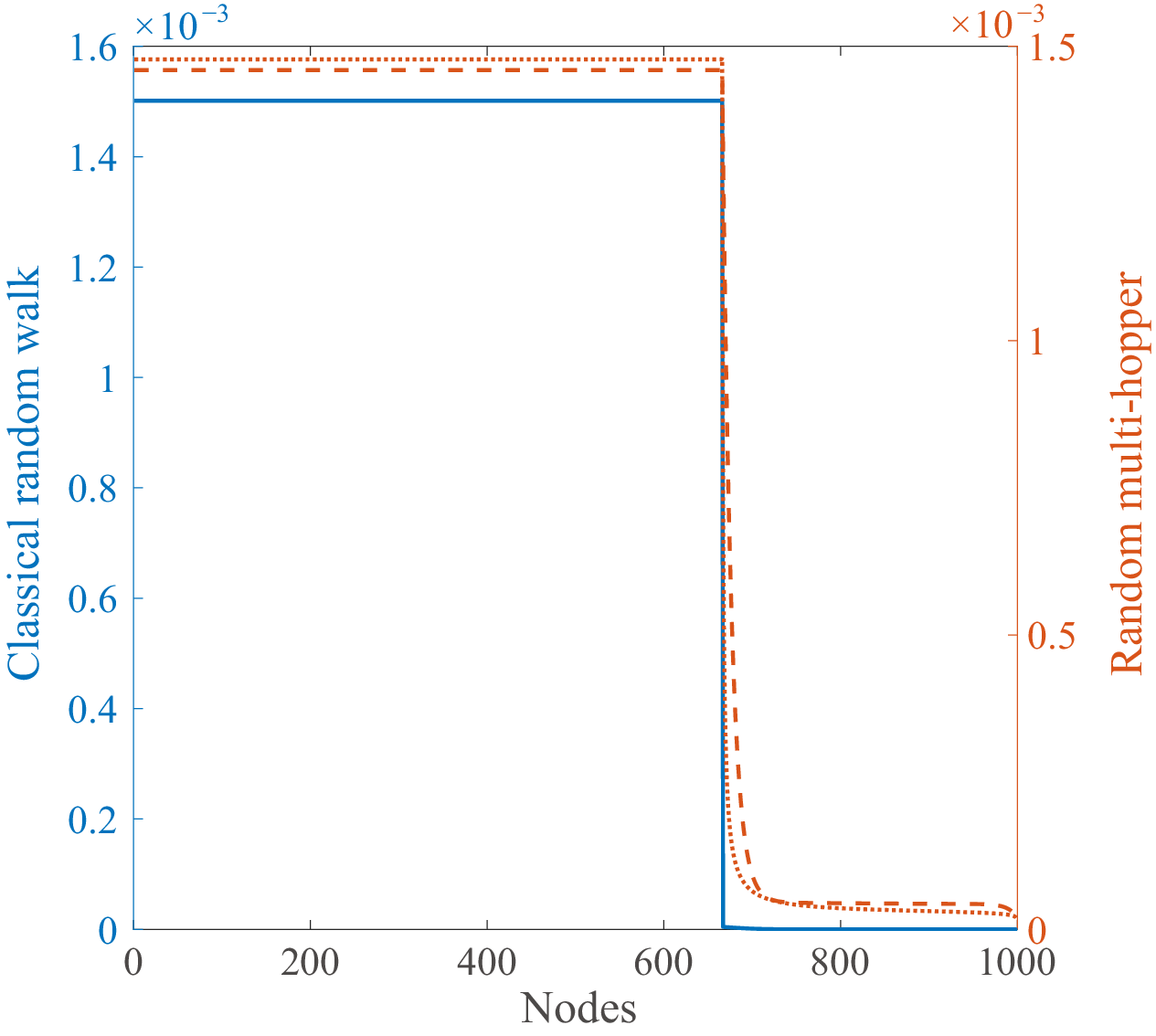}
\hfill
\includegraphics[width=0.32\textwidth]{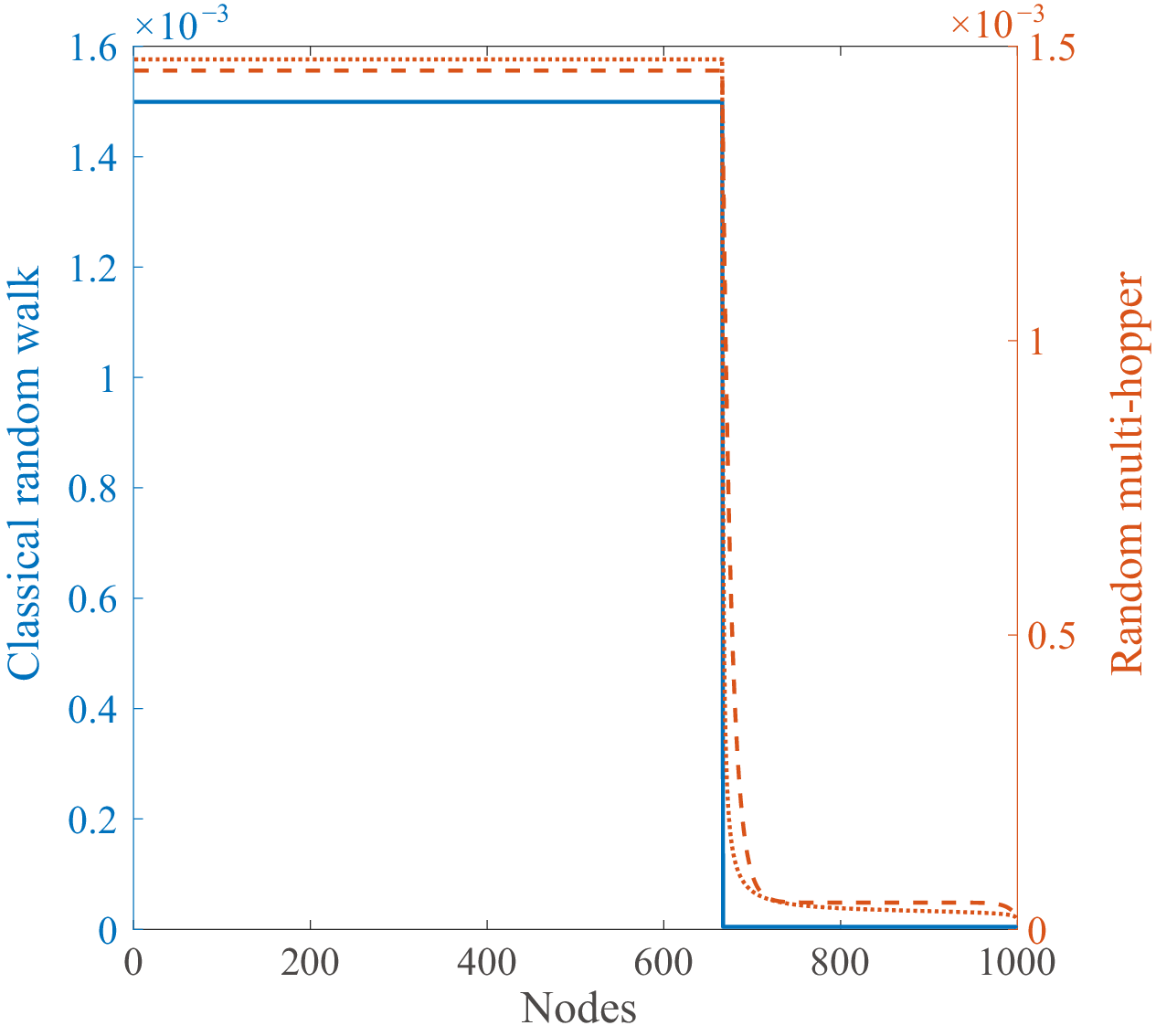}
\vspace{\baselineskip}
\\
\includegraphics[width=0.32\textwidth]{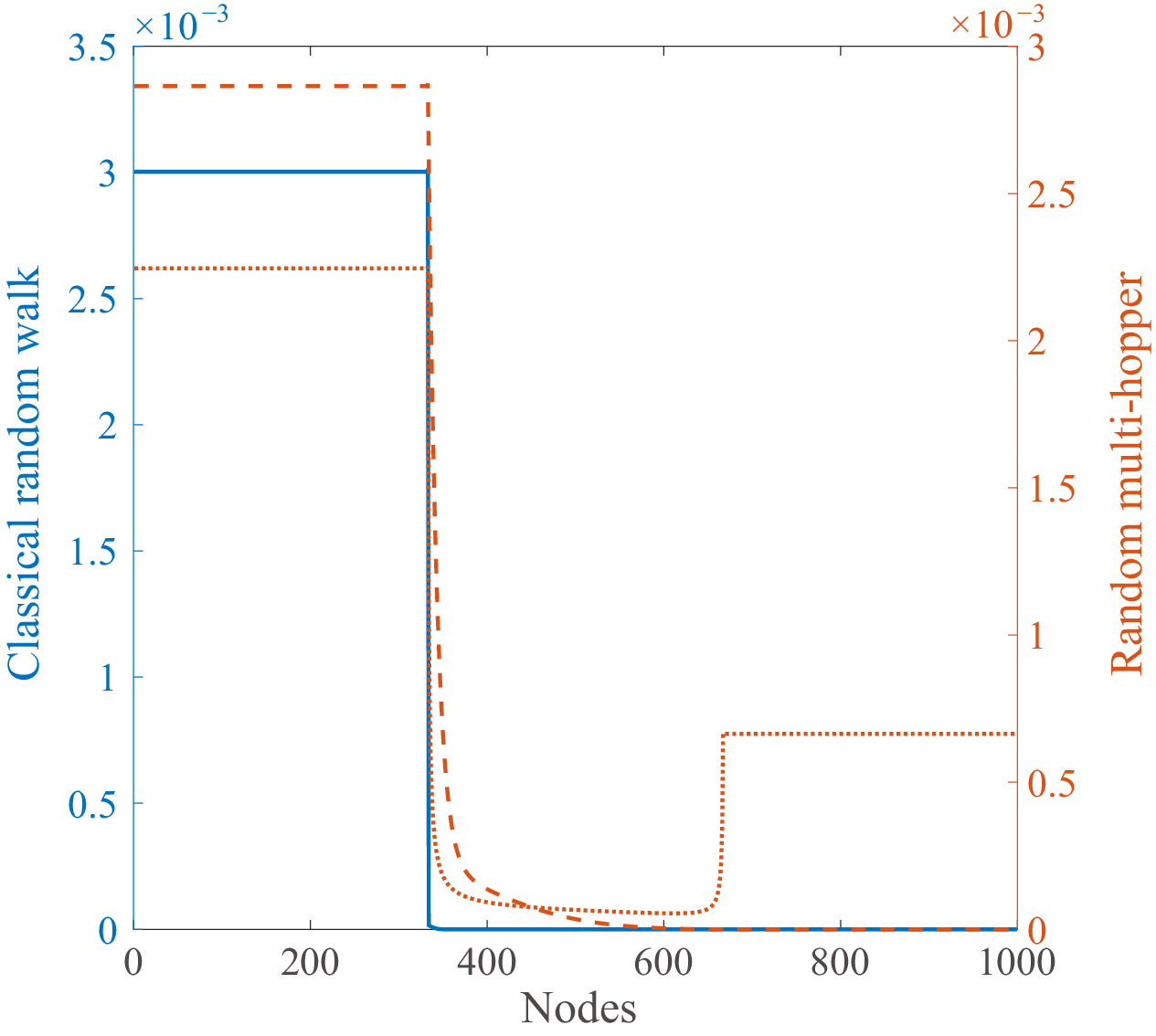}
\hfill
\includegraphics[width=0.32\textwidth]{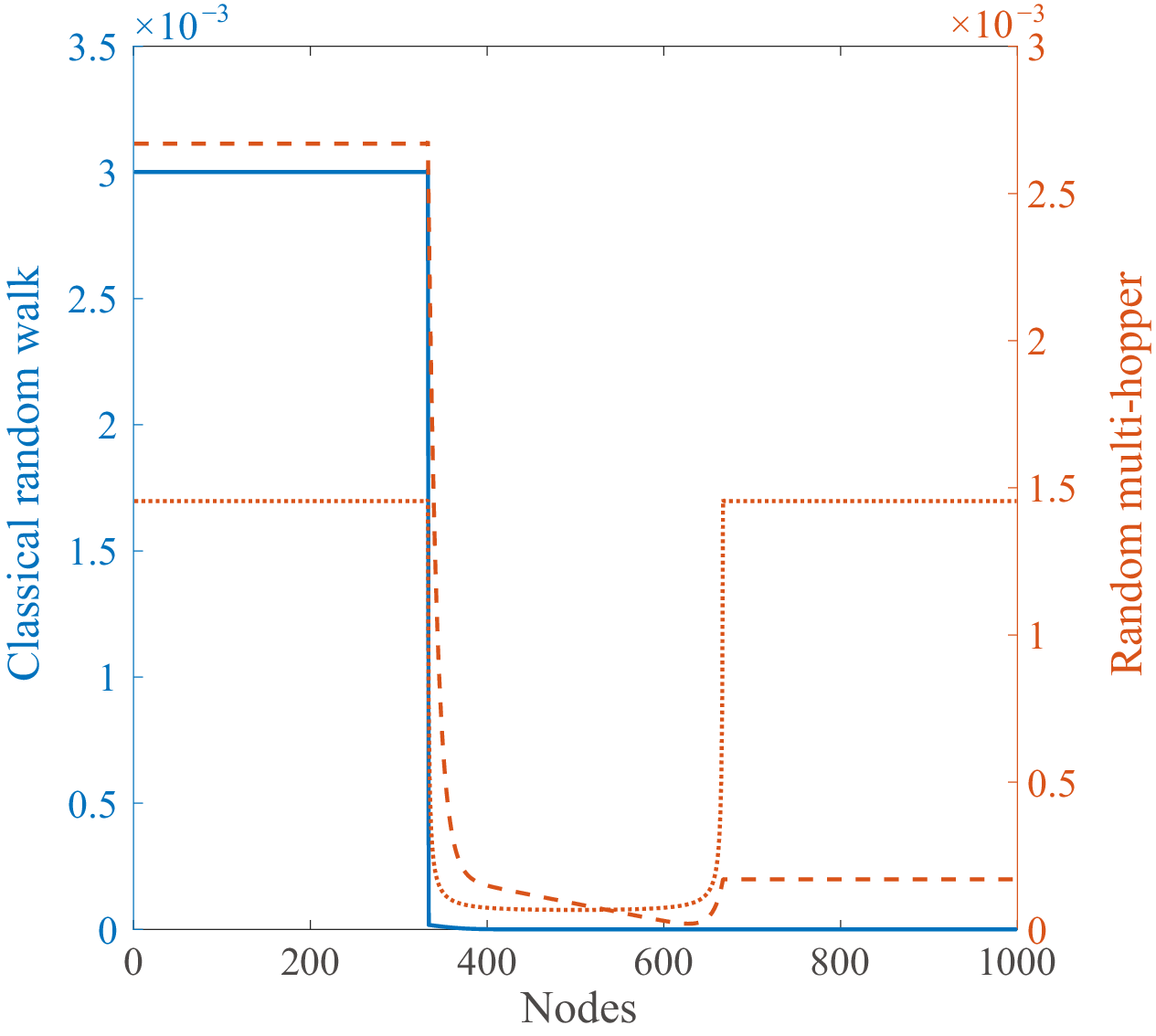}
\hfill
\includegraphics[width=0.32\textwidth]{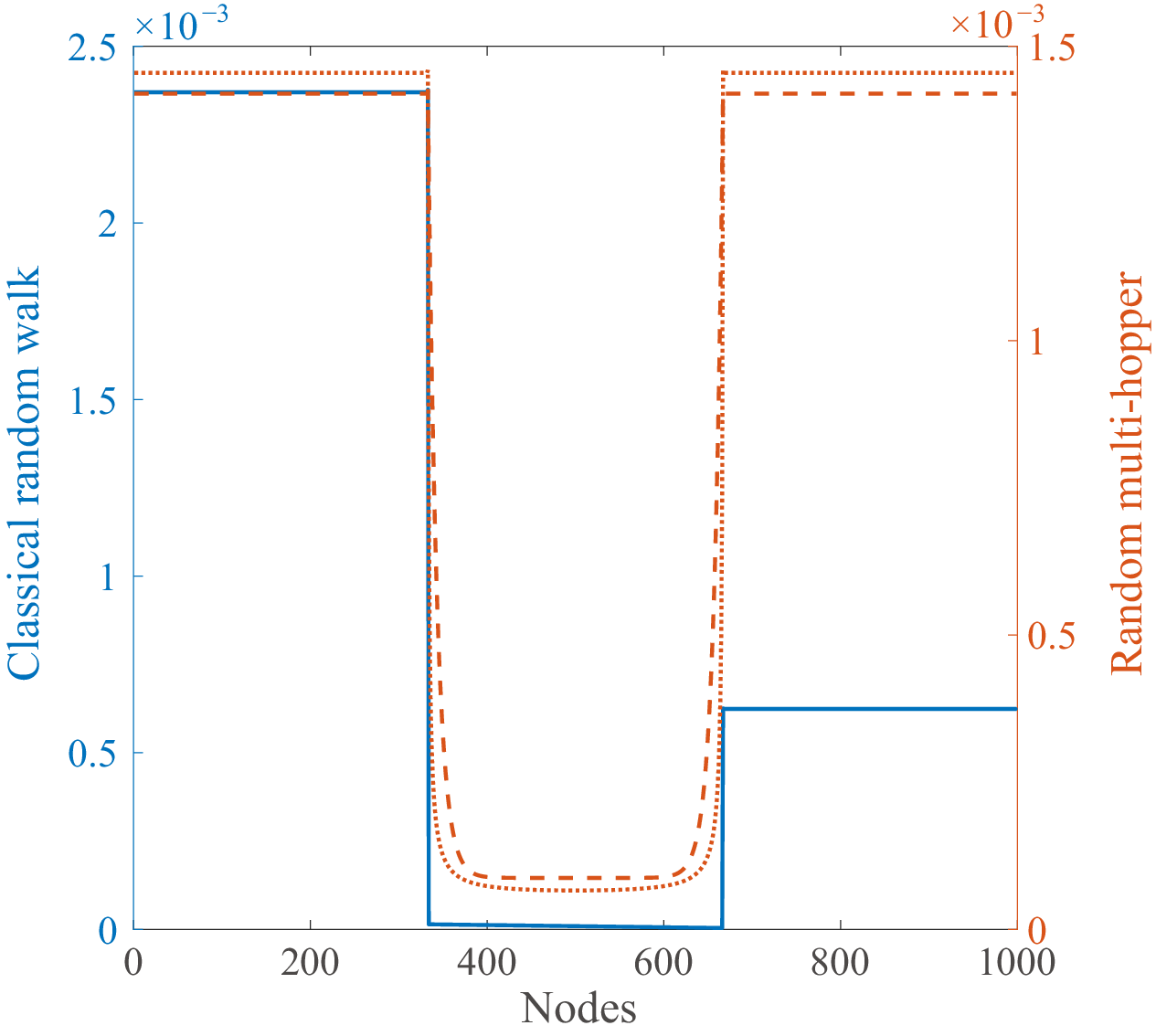}
\caption{Probability distribution at the different nodes of a lollipop $L\left(n,\left\lfloor \tfrac{2n}{3}\right\rfloor \right)$
(top) and barbell $B\left(n,\left\lfloor \tfrac{n}{3}\right\rfloor ,\left\lfloor \tfrac{n}{3}\right\rfloor \right)$
(bottom) graph with $n=999$ nodes. Classical random walk (blue solid
line) and the multi-hopper using the Mellin transform with $s=1$
(red broken line) and using the Laplace transform with $l=0.1$ (red
dotted line). The evolution of the probabilities is shown at three
different times for a random walk starting at node $i$ (see text)
at $t=50$ (left), $t=1000$ (center) and $t=10^{6}$ (right).}

\label{Barbell_time evolution-1}
\end{figure}

On the other hand, the random multi-hopper has a non-zero probability
of escaping directly from the clique at very short times. As can be seen in the right
panels of Fig.~\ref{Barbell_time evolution-1} even for the
small time $t=50$ the multi-hopper with Mellin transform has already
visited all the nodes of the graphs. For this short time, however,
the multi-hopper with Laplace transform has visited all the nodes
of the cliques plus the initial nodes of the path, but she has not
arrived yet at the node $j$. For time $t=1000$ the multi-hopper
with Mellin transform is already in the stationary state and the one
with Laplace transform has already visited all the nodes of the graphs.
At $t=10^{6}$ the multi-hopper has reached the stationary state for
both transforms. This significant difference with the classical RW
is due to the fact that the random multi-hopper is not trapped in
the cliques due to the fact that she can go directly from $i$ to
any node of the graph with a probability that decays as a function
of the distance from $i$. Then, the first few nodes of the path are
frequently visited by the multi-hopper as they are at relatively short
distances from the node $i$. Once, on these nodes the multi-hopper
can visit the most extreme nodes of the graphs in an easier way overtaking
the classical RW even at relatively short times. The way in which
a random multi-hopper is propagated through a path is analyzed in
the next subsection of this work.

\subsection{Path graphs}

Another interesting graph to consider is the path $P_{n}$. A path
$P_{n}$ is the graph having $n$ nodes all of degree 2 but two which
are of degree 1. As proved by Palacios~\cite{palacios2001resistance}
this graph has the maximum Kirchhoff index. For the normal RW, Palacios
proved that $\varOmega_\mathrm{tot}\sim n^{3}$. As the number of
edges in $P_{n}$ is $n-1$ one can easily see that $\left\langle H\left(P_{n}\right)\right\rangle \sim n^{2}$.
In Fig.~\ref{time evolution} we illustrate the evolution
of the probabilities of being at a given node of the path of 1000
nodes labelled in consecutive order from 1 to 1000, in which we have
placed the random walker at the node 1. As can be seen in the left panel of 
Fig.~\ref{time evolution} for $t=500$, the classical random walker
has visited only the first 100 nodes of the path while the random
multi-hoppers for both transforms has already visited all the nodes.
As the time increases, the random multi-hopper model gives almost
identical probabilities of finding the walker at any node of the path,
but the classical random walker still shows close to zero probability
of finding the walker at the other side of the path for times as large
as $t=5000$ (see the right panel of Fig.~\ref{time evolution}). 

\begin{figure}
\includegraphics[width=0.32\textwidth]{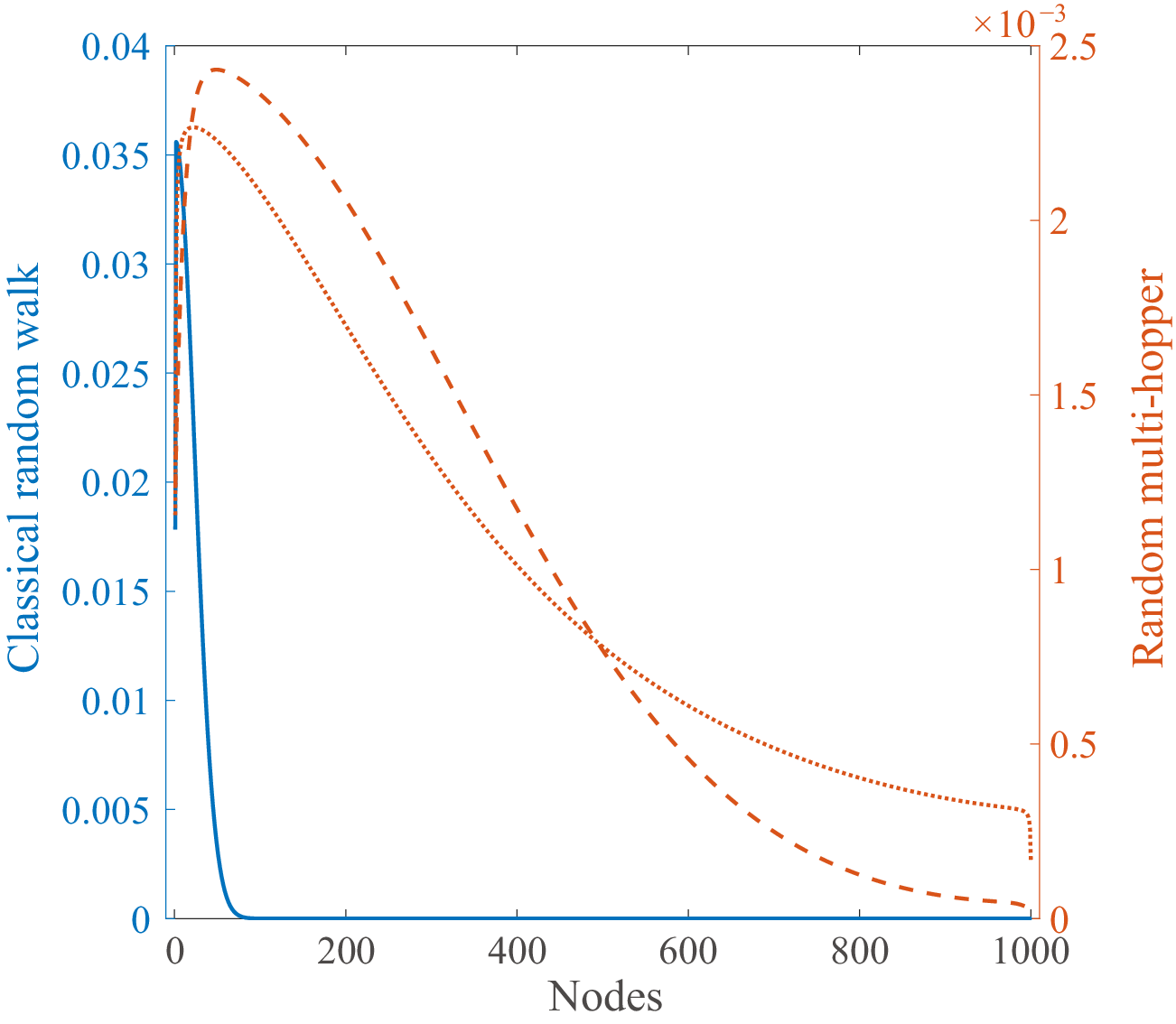}
\hfill
\includegraphics[width=0.32\textwidth]{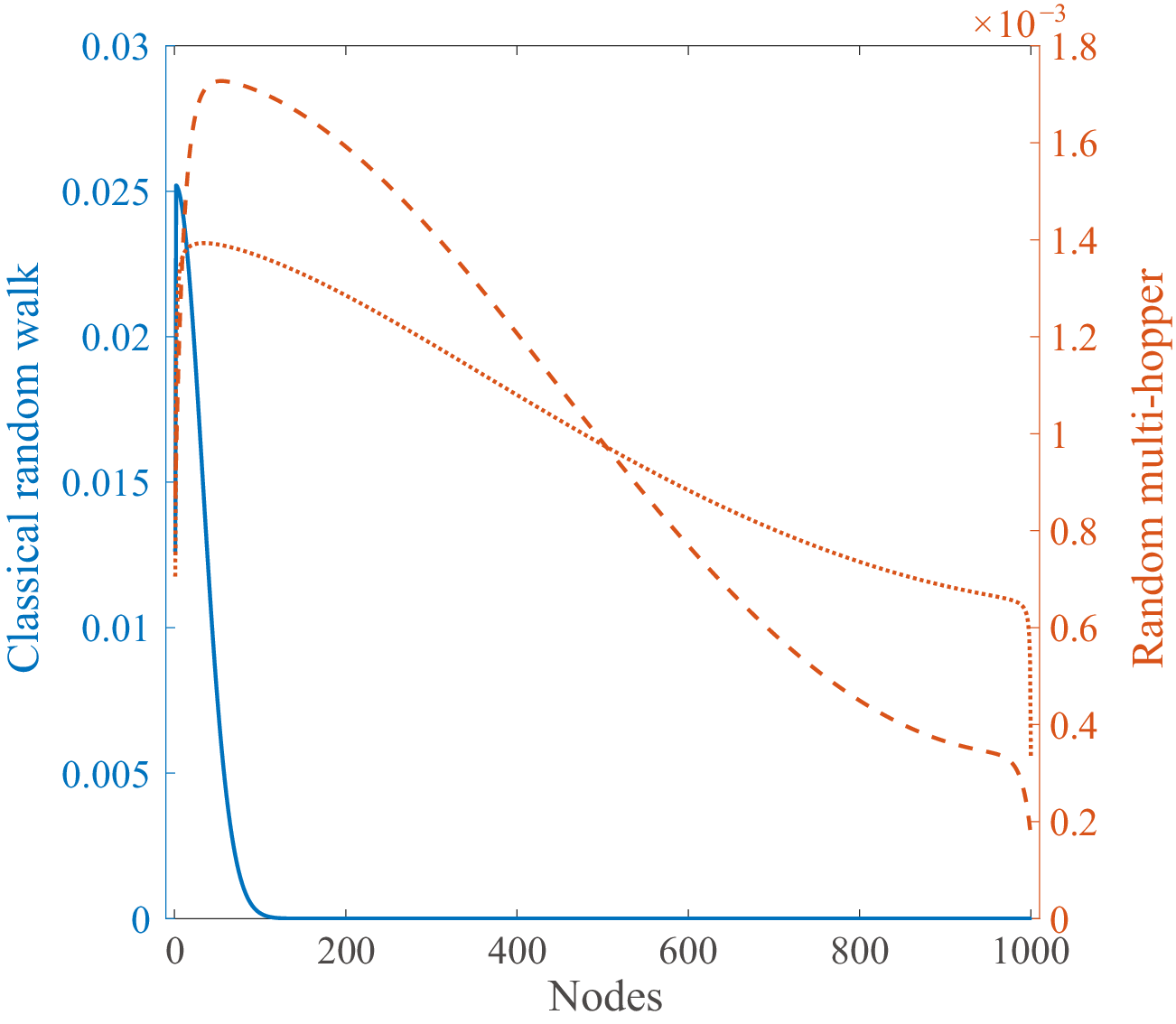}
\hfill
\includegraphics[width=0.32\textwidth]{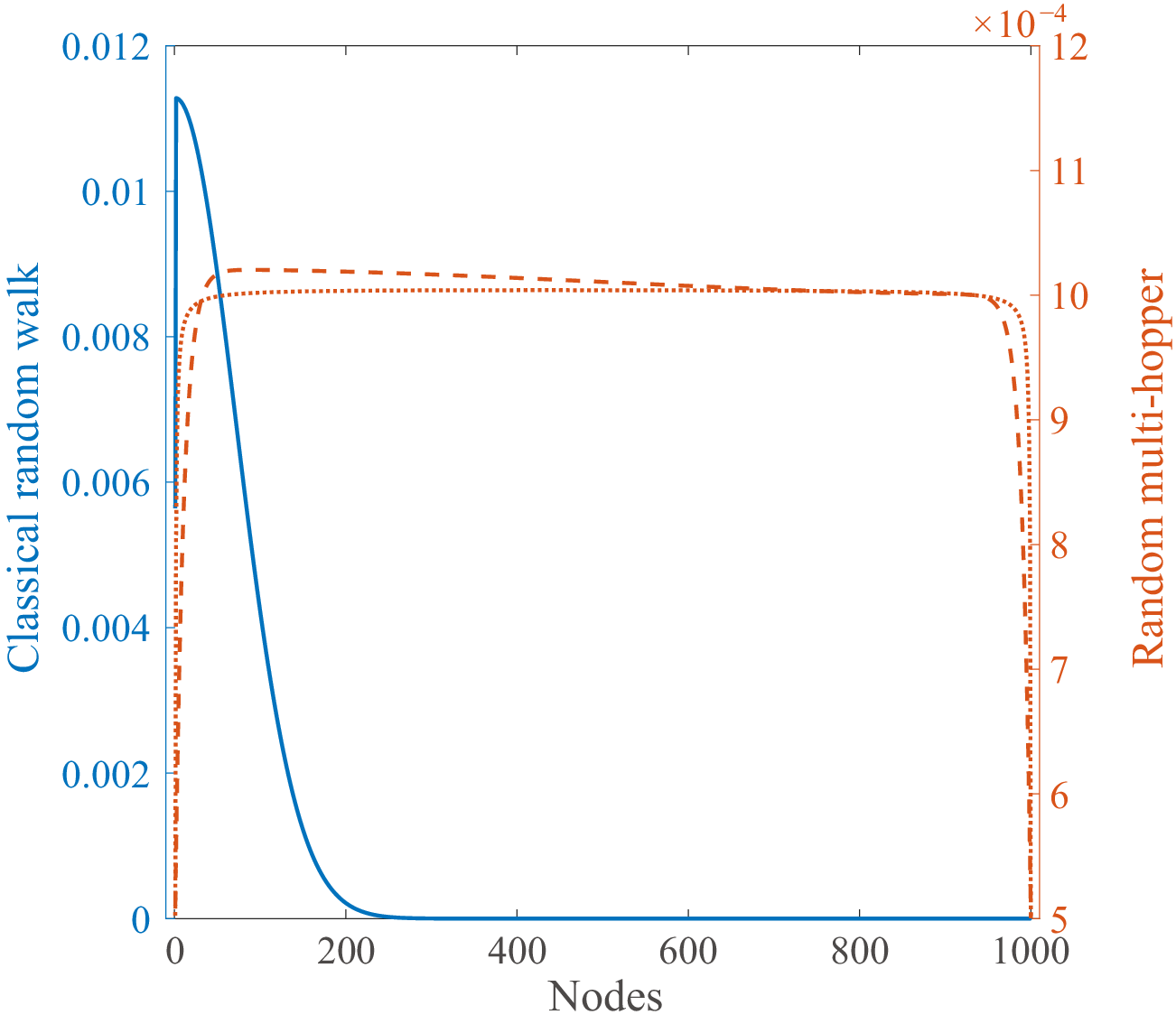}
\caption{Probabilities of finding the random walker at a given node of $P_{1000}$
at $t=500$ (left), $t=1000$ (center), and $t=5000$ (right) for
the classical (blue solid line) and multi-hopper random walk model
with Mellin transform with parameter $s=2$ and with the Laplace transform
with parameter $l=0.1$.}
\label{time evolution}
\end{figure}

As in the previous subsection we study here the influence of the graph
size over the hitting time in paths for both the Mellin and Laplace
transforms. In particular, we compare both transformations in the
multi-hopper random walk with the classical one for the path graph
with $100\leq n\leq1000$. As expected the average hitting time in
the classical random walk follows a quadratic dependence with the
number of nodes, $\langle H\rangle \approx0.3333n^{2}$.
However, for the multi-hopper one it follows power-laws with exponent
smaller than 2. For instance, for the dependence is of the form $\langle H^\tau\rangle \approx an^{b}$
with $b<2$.

The most interesting thing here is that as for the barbell and lollipop
graphs the average hitting time of paths also increases linearly with
the number of nodes for relatively small values of the parameters
$s$ and $l$. 

\subsection{Some remarks}

It is intuitive to think that the average shortest-path distance plays
a fundamental role in explaining the average hitting time of graphs
in the normal RW model. Then, because we allow for long-range jumps
in the multi-hopper model we would intuitively expect that such influence
of the shortest-path distance is diminished in this model. However,
one important thing that we have learn from the analysis of the lollipop,
barbell and path graphs is the following. Although the average shortest
path distance plays some role in the determination of the average
hitting times, it is the existence of large, relatively isolated,
clusters which plays the major role. That is, although in a path graph
we can have the largest possible average shortest path distance of
any graph with $n$ nodes---it has average path length equal to $\frac{n+1}{3}$---they
display average hitting time one order of magnitude smaller than the
lollipop and barbell graphs, which may have relatively small average
shortest path distances---particularly for the ones analyzed in this
section. This role of large clusters in graphs, which we discussed
in this section, is a great importance for the analysis of real-world
networks. Although these networks have relatively small average shortest-path
distance due to their small-world properties, they contain many communities---clusters
of tightly connected nodes, which are poorly connected among them---which
resemble the extremal situation of barbell and lollipop graphs.

\section{Random graphs}

In this section we explore the multi-hopper model for two types of
random networks: Barab\'{a}si-Albert (BA)~\cite{barabasi1999emergence}
and Erd\H{o}s-R\'{e}nyi (ER)~\cite{erdos1959random} types. We use the
exact result obtained for the expected hitting time averaged over
all pairs of nodes $\langle \hat{H}^{\tau}\rangle $ in Eq.~(\ref{eq3.14}).

The analysis of the Barab\'{a}si-Albert (BA) and Erd\H{o}s-R\'{e}nyi (ER)
random networks shows that the hitting time increases linearly with
the number of nodes in the graph (data not shown). In all cases
the use of the Mellin transform in the multi-hopper drops the slope
of the lines $\langle \hat{H}^{\tau}\rangle \approx an+b$
expressing the dependence
of the hitting time with the number of nodes. Thus, we found that
in general, using the long-range strategies, the resulting random
walker reaches more efficiently any site on the network in comparison
with the normal random walk. 

In Fig.~\ref{Figure_5_2} we fixed the number of nodes to $n=2000$
in ER and BA networks and we calculate the average hitting time as
a function of the parameters $l$ and $s$ for the Laplace and Mellin
transform, respectively. The results confirm previous findings~\cite{riascos2012long}
that in the limit for $l,s\to0$ the hitting times reach the value
$n-1$. However, for parameters in the interval $(0,10)$, we see
how the two types of strategies present a strong variation in the
average hitting time  $\langle \hat{H}^{\tau}\rangle $.
This is a direct consequence of how the random walk strategies assign
weights to small, intermediate and large steps. Finally, for large
values of the parameters, long-range transitions appear with low probability
and the values of $\langle \hat{H}^{\tau}\rangle $ are
equal to the results for the normal random walk strategy with transitions
only to nearest neighbors. These results are of great significance
for the further analysis of real-world networks in the next section
of this work.
\begin{figure}
\includegraphics[width=0.45\textwidth]{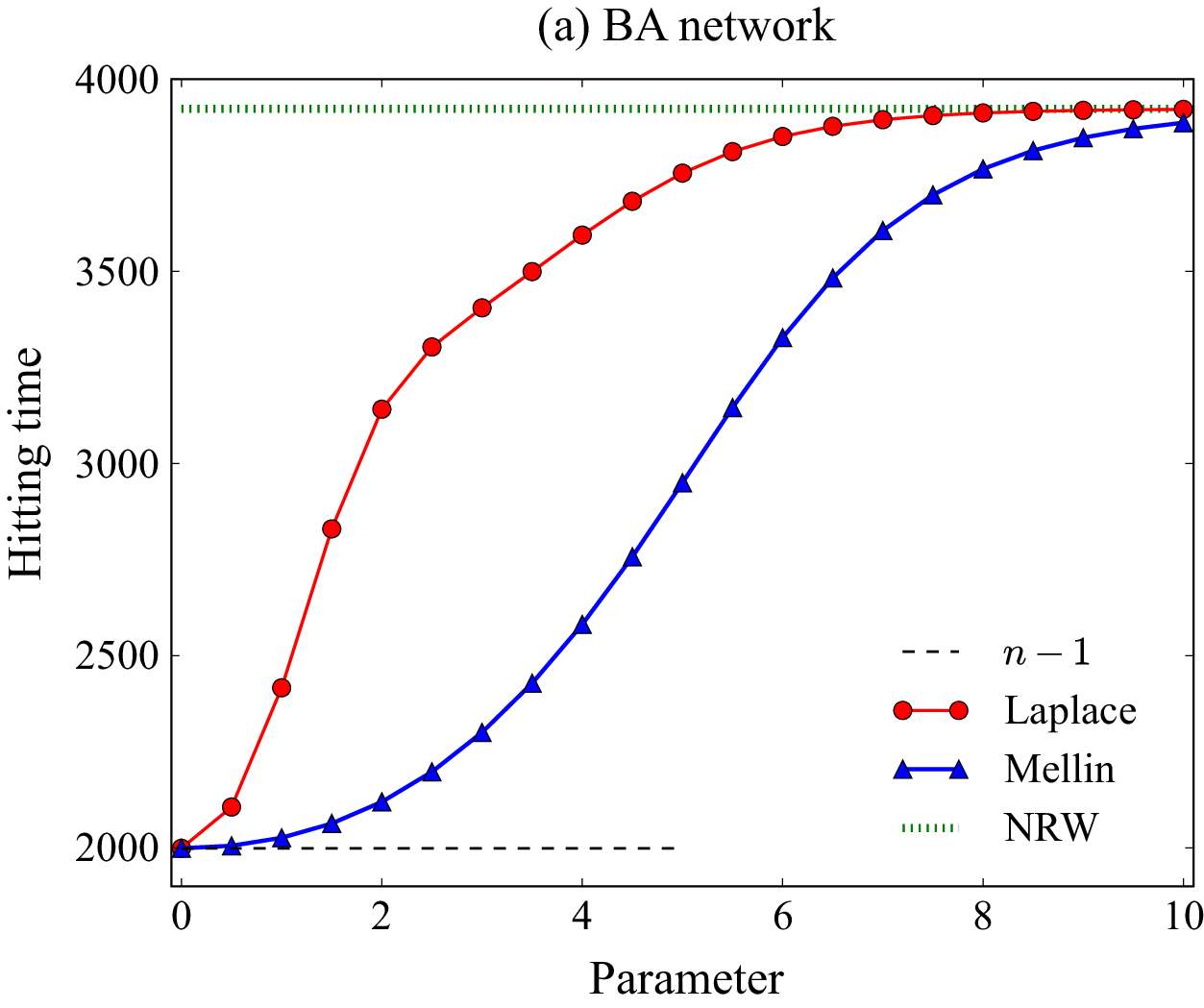}
\hfill
\includegraphics[width=0.45\textwidth]{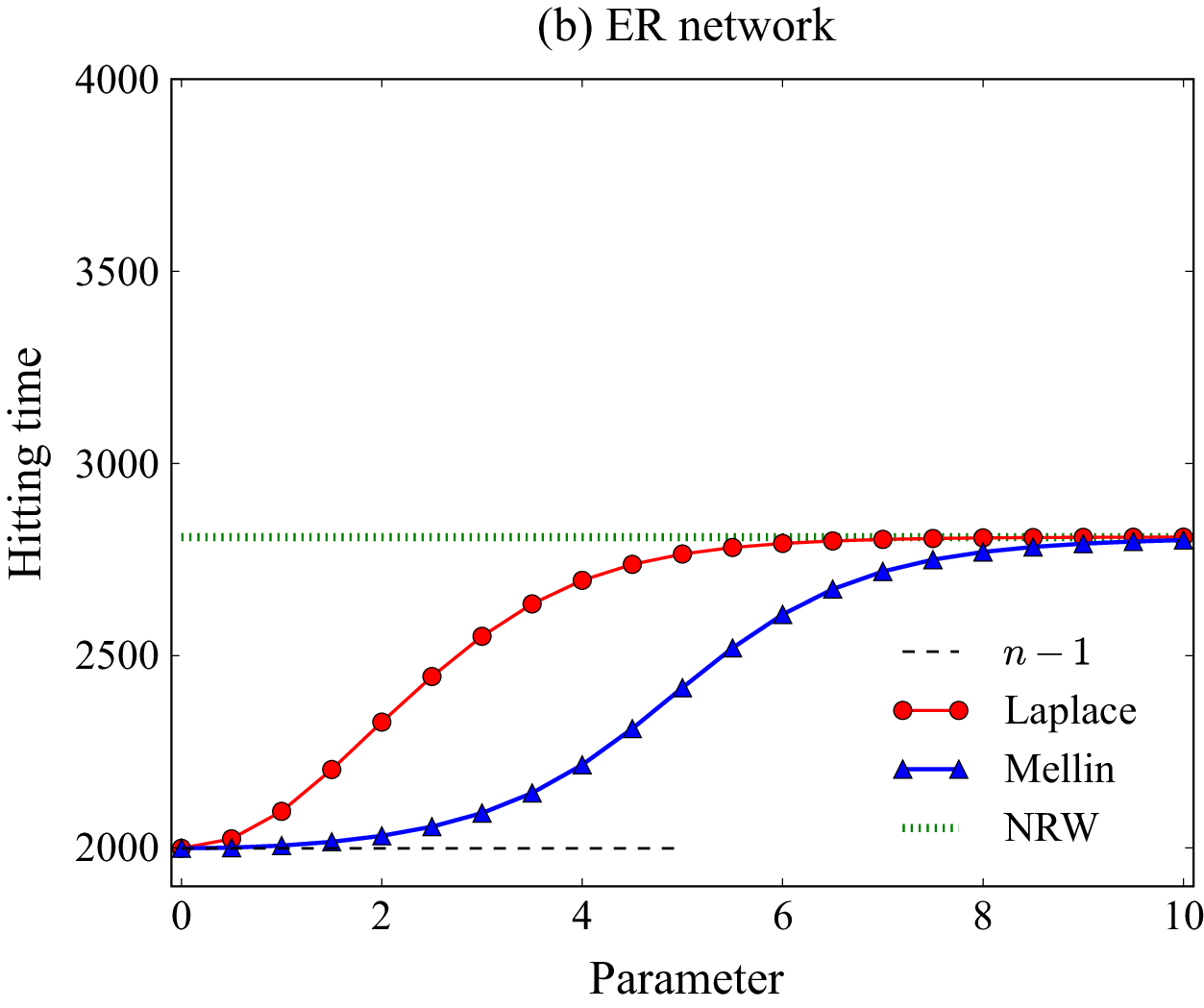}
\caption{\label{Figure_5_2} Influence of the degree distribution on the average
hitting time for random networks with $n=2000$ nodes. (a) Barab\'{a}si-Albert
(BA) network, (b) connected Erd\H{o}s-R\'{e}nyi (ER) networks with probability
$p=\log(n)/n$. We plot the results for the hitting time as a function
of the parameter $l$ for the Laplace transform and $s$ for the Mellin
transformation. We depict, with dashed lines, the results for the
normal random walk (NRW) and $n-1$ obtained for a complete graph. }
\end{figure}

As we have observed in the previous analysis there are very significant
differences between the random networks considered here and the graphs
analyzed in the previous section, where the average hitting time increases
as a third or second power of the number of nodes. The linear increase
observed here for the random graphs studied cannot be understood only
on the basis of the fact that they display relatively small shortest-path
distances. For instance, we can construct barbell graphs $B\left(n,\left\lfloor \tfrac{n-k}{2}\right\rfloor ,\left\lfloor \tfrac{n-k}{2}\right\rfloor \right)$
with small values of $k$, which have small average shortest-path
distance.  A graph $B\left(n,\left\lfloor \tfrac{n-k}{2}\right\rfloor ,\left\lfloor \tfrac{n-k}{2}\right\rfloor \right)$
has only distances $d_{ij}\in\left[1,k\right]$. One important difference,
however, between the studied random networks and the barbell and lollipop
graphs previously considered is the lack of large cliques in these
random graphs which may trap the random walker inside them. In the
next section we study this problem by using random graphs with different
intercommunity density of links. In addition, we study the influence
of the degree distribution on the hitting time of these random graphs
with the goal of understanding the differences between Erd\H{o}s-R\'{e}nyi
and Barab\'{a}si-Albert networks.

\subsection{Influence of communities and degree distribution}

We start here by considering the influence of the presence of clusters
of nodes defined in the following way. Let us consider a network with
$n$ nodes. Let us make a partition of the network in $k$ clusters
of size $\left\lfloor \frac{n}{k}\right\rfloor $. Let $C_{i}$ and
$C_{j}$ be two of such clusters. Then, the probability that two nodes
$p,q\in C_{i}$ are connected is much larger than the probability
that two nodes $r\in C_{i}$ and $s\in C_{j}$ are connected. This
gives rise to higher internal densities of links in the clusters than
the inter-cluster density of links. It is a well-known fact that neither
the Erd\H{o}s-R\'{e}nyi nor the Barab\'{a}si-Albert networks contain such
kind of clusters. The lack of such clusters---known in network theory
as communities---is characterized by the so-called good expansion
properties of these graphs. Loosely speaking a graph is an expander
if it does not contain any structural bottleneck, i.e., a small group
of nodes or edges whose removal separates the network into two connected
components of approximately the same size~\cite{hoory2006expander}.
We remark here that both ER and BA graphs have been proved to be expanders
when the number of nodes is very large~\cite{hoory2006expander,mihail2003certain}.
Then, we use here an implementation of the algorithm described by
Lancichinetti \textit{et al}.~\cite{lancichinetti2008benchmark} to produce
undirected random networks with communities with a fixed average degree
$\langle k\rangle$. A mixing parameter $\mu$ defines the fraction
of links that a node share with nodes in other communities. A small
value of the mixing parameter produces graphs with tightly connected
clusters which are poorly connected among them. That is, it produces
very well-defined communities in the graph. As the mixing parameter
approaches the value of one, the communities disappear and the graph
looks more and more as an expander for sufficiently large number of
nodes. 

Here, we explore the effect of communities in the capacity of the
multi-hopper random walk strategy to reach any site of the network
by constructing random graphs with the same number of nodes and edges
but changing the mixing parameter. In Fig.~\ref{Figure_4} we depict
the average hitting time $\langle H^{\tau}\rangle$ for different
values of the parameters $l$ and $s$ for networks with communities
constructed as described before. As can be seen here the
random graphs with well defined communities, i.e., small values of
the mixing parameter, makes that the random walker takes significantly
longer time to explore the whole network. This is particularly true
for relatively large values of the Mellin and Laplace parameters of
the multi-hopper model, which indicate that the normal RW is significantly
less efficient in networks having communities than in networks not
displaying such structural characteristic. Here again, as these parameters
approach zero the hitting time decays to the lower bound as expected
from the theory. In closing, the small hitting times observed for
the random graphs studied in the previous subsection are mainly due
to the fact that these graphs are expanders and they lack any community
structure, which may trap the random walker for longer times without
visiting other clusters. In those cases analyzed here where there are
communities, the multi-hopper solves this trapping problem due to
the fact that it is allowed to jump from one community to another,
reducing her time inside each of the clusters visited. 
\begin{figure}
\includegraphics[width=0.45\textwidth]{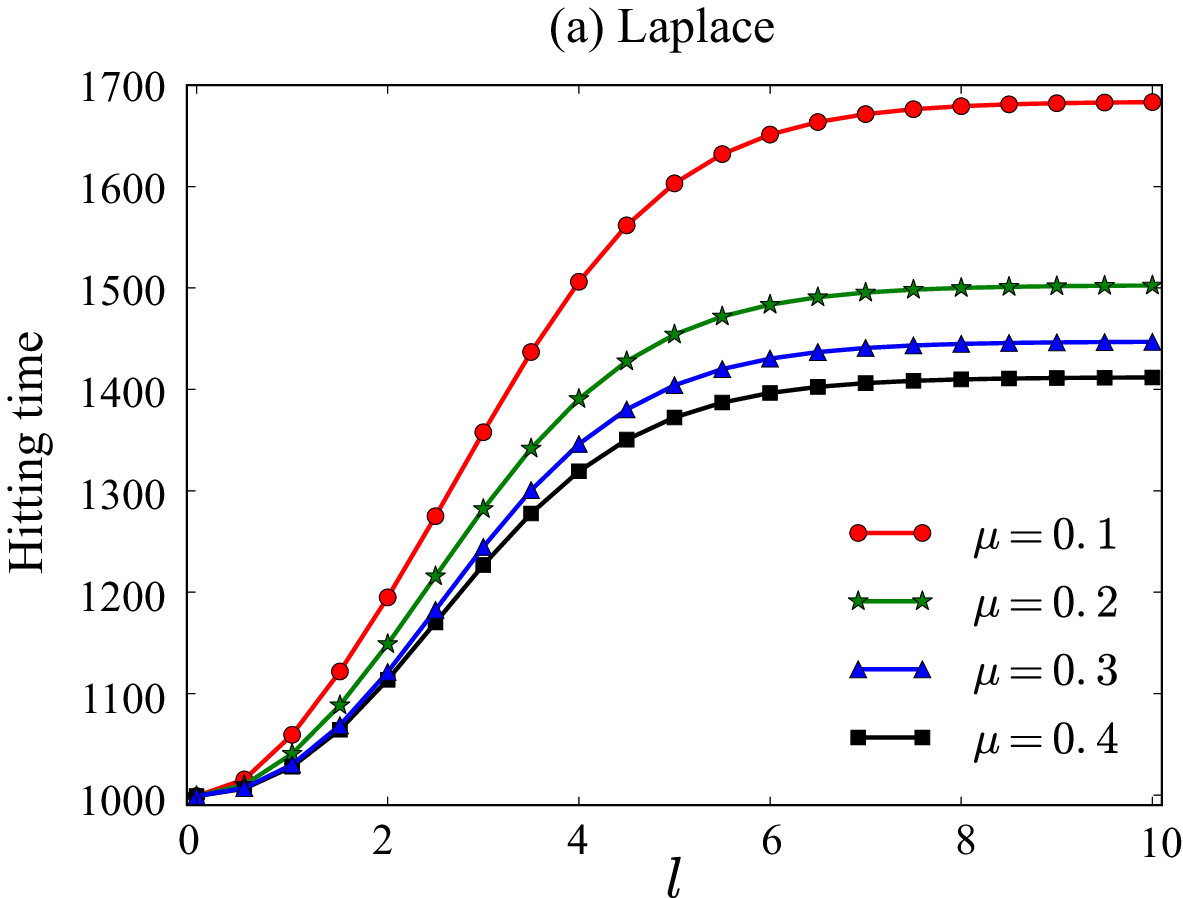}
\hfill
\includegraphics[width=0.45\textwidth]{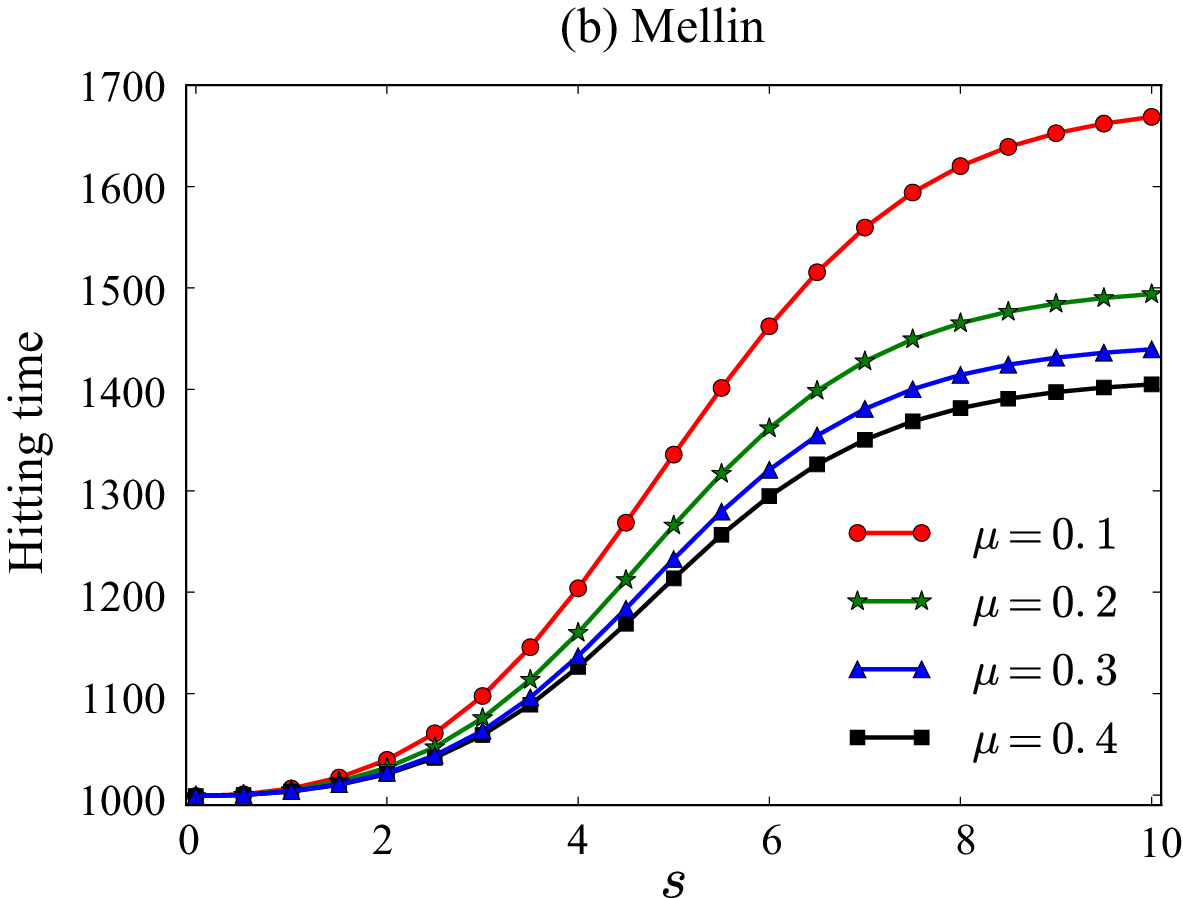}
 \caption{\label{Figure_4} Average hitting time for the multi hopper random
walker in networks with communities. (a) Laplace and (b) Mellin strategies.
We explore networks with $n=1000$ nodes, an average degree $\langle k\rangle=15$
and different values of the mixing parameter $\mu$ that defines the
fraction of connection that a node has with nodes in other communities. }
\end{figure}

Now we move to the consideration of the influence of the degree distribution
on the performance of the random multi-hopper. We then study the stationary
probability distribution $\pi^{\tau}(i)$ for the Laplace and Mellin
transformation in a Barab\'{a}si-Albert and an Erd\H{o}s-R\'{e}nyi network,
with $n=2000$ nodes. The results are obtained from the calculation
of the long-range degrees $k^{\tau}(i)$ in Eqs.~(\ref{LRDegreeM})--(\ref{LRDegreeE})
and the respective normalization defined in Eq.~(\ref{StatDistf}). 

In Fig.~\ref{Hitting_degree} we resume our results. The important
aspect of these plot is to consider the slope of the corresponding
curves for different values of the Mellin and Laplace transforms in
the multi-hopper model. If we compare the slopes for the ER network
with that of the BA one, we observe that the first is smaller and
closer to the constant line $n-1$ than the second. The smallest hitting
time is obtained when the slope coincides with this line, which represents
the fully connected graph. Thus, the ER graphs are already close
to this slope and this is the main reason why they display relatively
small hitting times. However, in the BA model when $s,l$ are very
large the slope of the curves are very steep and far away from the
asymptotic result. As soon as these parameters approach zero the slope
of the curves become more flat approaching $\pi^{\tau}(i)=n-1$ as
a consequence of the fact that the graph approaches the fully connected
one. That is, the long-range dynamics changes the way in which the
random walker reaches the nodes. For small values of the parameter
$s$ or $l$, the stationary probability distribution reaches the
value $\pi(i)=1/n$. On the other hand, the inverse of the stationary
probability distribution defines the average time $\left\langle t^{\tau}(i)\right\rangle =\frac{1}{\pi^{\tau}(i)}$
needed for the random walker to return for the first time to the node
$i$. In this way, the random walker returns to sites with high values
of $\boldsymbol{\pi}^{\tau}$ and, gets trapped in regions with this property.
As we can see in Fig.~\ref{Hitting_degree}, the effect of the long-range
strategies is to reduce the probability to revisit sites highly connected
and increasing the capacity to reach any site of the network. 
\begin{figure}
\includegraphics[width=\textwidth]{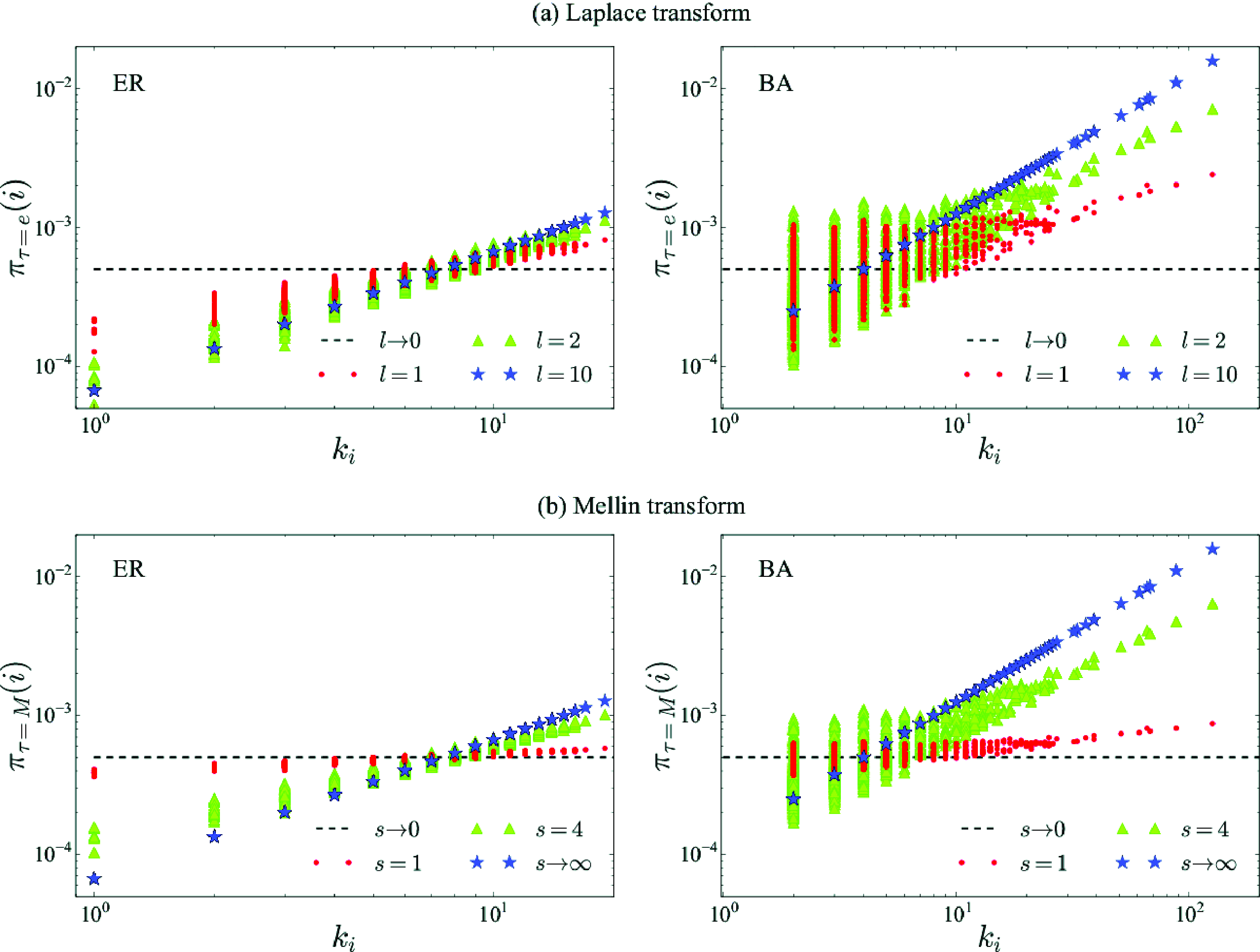}
\caption{Stationary probability distribution for multi-hopper random walkers
on Barab\'{a}si-Albert and Erd\H{o}s-R\'{e}nyi networks with $n=2000$ nodes.
(a) Laplace transform, (b) Mellin transform. }
\label{Hitting_degree}
\end{figure}

In closing, in this section we have seen that a random walker can
be trapped in certain regions of a network---i.e., having larger probability
of staying at these regions than in other parts of the graph--- due
to two different factors. The first is the presence of clusters of
highly connected nodes in which the random walker is retained for
long times before she visits other clusters of the graph. The second
is the existence of hubs---highly connected nodes---which make that
the random walker returns frequently to them making her exploration
of the network more difficult. These two characteristics, the presence
of communities and the existence of fat-tailed degree distributions,
are well known to be ubiquitous in real-world networks. Then, the
observation that the random multi-hopper overcome both of these traps
make this model an important election for the exploration of real-world
networks, which is the topic of our next section.

\section{Real-world networks}

One of the areas in which the random multi-hopper can show many potential
applications is in the study of large real-world networks. Normal
random walks on networks have been previously used as mechanisms of
transport and search on networks~\cite{adamic2001search,guimera2002optimal,noh2004random}.
These are graphs representing the networked skeleton of complex systems
ranging from infrastructural and technological to biological and social
systems. As an example of the potential applications of the random
multi-hopper for these systems we consider the exploration of the
electrical power grid of the western USA. In this case we compare
the normal RW with the multi-hopper by placing the walker at the node
having the largest closeness centrality in the network. This is the
node which is relatively closer to the rest of the nodes of the graph.
We compare the results with the selection of the initial node as the
ones having the smallest closeness centrality among all the nodes
in the graph, i.e., the one relatively farthest from all the other
nodes. In Fig.~\ref{power_grid} we illustrate the evolution of
the probability of finding the random walker at a given node of the
power grid at $t=50$, $t=500$, and $t=5000$. When the initial node
is the one with the largest closeness centrality, at relatively short
times (see the top-left panel of Fig.~\ref{power_grid}) the normal random
walk has left some regions of the power grid unexplored. This situation
is more critical when the initial node is the one with the smallest
closeness centrality (see the bottom-left panel of Fig.~\ref{power_grid}),
where the graph remains almost totally unexplored. In this last case
even when $t=500$ there are vast regions of the power grid that have
not been visited by the walker. On the contrary, the random multi-hopper
reaches the stationary state at very early times and at $t=50$ she
already has visited every node of the power grid independently of
whether the initial node has the largest or the smallest closeness
centrality. 

\begin{figure}
\includegraphics[width=\textwidth]{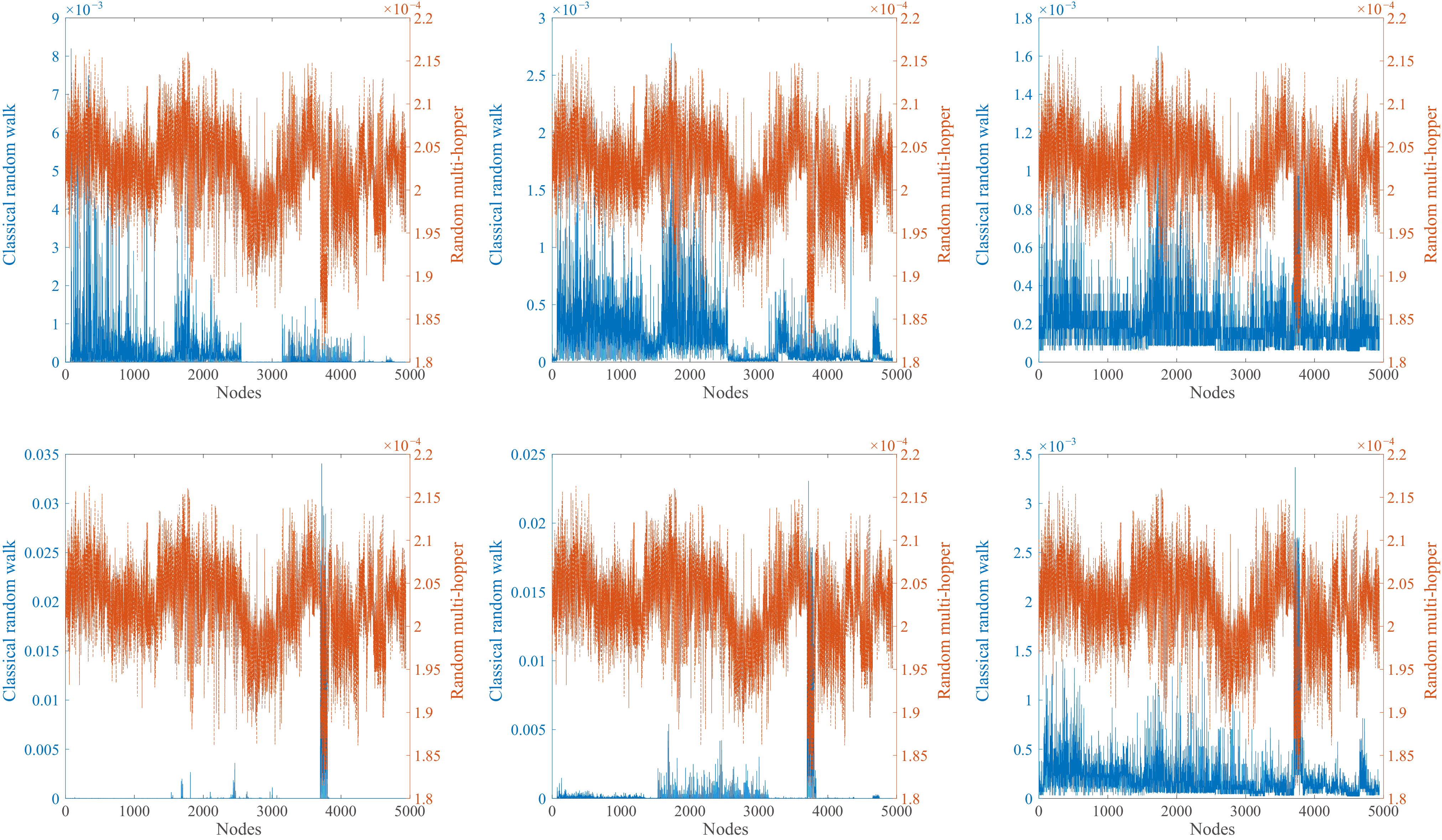}
\caption{Probability of finding the random walker at a given node of the western
USA power grid for the classical random walker (blue solid line) and
for the multi-hopper with $s=0.1$ (red broken line) and $l=0.01$
(red dotted line). The random walker started her walk at the node
with the largest (top panels) and the smallest closeness centrality
(bottom panels). The snapshots are taken at three different times,
at $t=50$ (left), $t=500$ (center), and $t=5000$ (right).}
\label{power_grid}
\end{figure}

We finally study a few networks representing a variety of real-world
complex systems, including biological, communication and infrastructural
ones. In Table~\ref{real world} we report the sizes of these networks
as well as the average hitting times using the normal random walk and the
multi-hopper with Mellin transform. By using the expression~(\ref{eq:s_bound})
we can estimate the lower bound for the value of the Mellin parameter
$s$ for which $\langle \hat{H}^{\textnormal{Mell}}\rangle \leq n$. These
values are given in Table~\ref{real world} as $s_{lower}$ for all
the networks studied in this section. In addition we calculate the
actual value of this parameter for which $\langle \hat{H}^{\textnormal{Mell}}\rangle \leq n$
in these networks and report it as $s_{c}$ in Table~\ref{real world}.
The values of $s_{c}$ are obtained as follow. We calculate the value
of $\langle \hat{H}^{\textnormal{Mell}}\rangle $ for different values
of $s$ and obtain a fit of the form: $\langle \hat{H}^{\textnormal{Mell}}\rangle \approx\alpha s^{2}+\beta$
for $0.01\leq s\leq0.5$. Obviously, $\beta=n-1$, which is the lowest
value obtained by $\langle \hat{H}^{\textnormal{Mell}}\rangle $ for any
graph. Using these fitting equations we then calculate the values
of $s_{c}$ reported in Table~\ref{real world}. As can be seen the
values of $s_{c}$ are as average 10 times larger than the lower bound
expected from the lollipop graphs of the same size as the studied
networks. 

\begin{table}
\begin{center}
\begin{tabular}{|>{\centering}p{3cm}|>{\centering}p{1.5cm}|>{\centering}p{1.5cm}|>{\centering}p{1.5cm}|>{\centering}p{1.5cm}|>{\centering}p{1.2cm}|}
\hline 
network & $n$ & $\left\langle H\right\rangle $ & $s_{c}$ & $s_{B}$ & \% impr.\tabularnewline
\hline 
\hline 
Bio\_PPI\_yeast & 2,224 & 8,652 & 0.145 & 0.0128 & 389\tabularnewline
\hline 
City\_Atlanta & 3,234 & 11,973 & 0.088 & 0.0106 & 370\tabularnewline
\hline 
Colab\_Geom & 3,621 & 15,719 & 0.086 & 0.0100 & 434\tabularnewline
\hline 
City\_Berlin & 4,495 & 13,752 & 0.094 & 0.0090 & 306\tabularnewline
\hline 
Power\_USA & 4,941 & 34,455 & 0.098 & 0.0086 & 697\tabularnewline
\hline 
City\_Barcelona & 5,575 & 17,282 & 0.063 & 0.0081 & 310\tabularnewline
\hline 
Colab\_AstroPh & 17,903 & 101,072 & 0.048 & 0.0045 & 565\tabularnewline
\hline 
City\_Seattle & 20,207 & 108,746 & 0.041 & 0.0042 & 538\tabularnewline
\hline 
Colab\_CondMat & 21,363 & 77,298 & 0.049 & 0.0041 & 362\tabularnewline
\hline 
Comm\_Enron & 33,696 & 193,714 & 0.040 & 0.0033 & 575\tabularnewline
\hline 
\end{tabular}
\end{center}
\caption{Real-world networks studied in this work, their number of nodes $n$
and the average hitting time of the normal random walker $\left\langle H\right\rangle $.
$s_{c}$ is the value of the Mellin parameter $s$
for which the corresponding network has hitting time smaller than
$n$. The value of $s_{B}$ is the Mellin parameter
$s$ for which the corresponding a lollipop graph $L\left(n,\left\lfloor \frac{3n}{2}\right\rfloor \right)$
with the same number of nodes as the real-world network has hitting
time smaller than $n$. The last column, \% improv.,
represents the percentage of improvement in the hitting time using
the Mellin-transformed multi-hopper respect to the NRW.}
\label{real world}
\end{table}

\section{Conclusions}

We develop here a mathematical and computational framework for using
random walks with long-range jumps on graphs of any topology. This
multi-hopper model allows a random walker positioned at a given node
of a simple, connected graph to jump to any other node of the graph
with a probability that decays as a function of the shortest-path
distance between her original and final positions. The decaying probabilities
for long-range jumps are selected as Laplace or Mellin transforms
of the shortest-path distances in this work. We prove here that when
the parameters of these transforms approach asymptotically zero,
the hitting time in the multi-hopper approaches the minimum possible
value for a normal random walker. Thus, the multi-hopper represents
a super-fast random walker hopping among the nodes of a graph. We
show by computational experiments that the multi-hopper overcomes
several of the difficulties that a normal random walker has to explore
a graph. For instance, the multi-hopper explores more efficiently
a graph having clusters of highly interconnected nodes, which are
poorly connected to other clusters, i.e., the presence of network
communities, than the normal random walker. It also overcomes the
normal random walker in those graphs with very skew degree distributions,
such as scale-free networks. In these graphs, the normal random walker
visits more frequently the hubs than the nodes of low degree, thus
getting stacked around the high-degree nodes of the graph. Finally,
we illustrate how the multi-hopper can be useful for transport and
search problems in real-world networks where these structural heterogeneities,
i.e., presence of communities and skew degree distributions, are more
a rule than an exception. We hope that the use of the random multi-hopper
will open new avenues in the exploration of lattices, graphs and real-world
networks.

\section*{Acknowledgment}

EE thanks the Royal Society for a Wolfson Research Merit Award.
NH acknowledges support of JSPS KAKENHI Grant Number JP15K05200.
JCD acknowledges support from: the FRS-FRNS (National Fund of Scientific Research); the Interuniversity Attraction Poles Programme DYSCO (Dynamical Systems, Control, and Optimisation), initiated by the BELSPO (Belgian Science Policy Office); and the ARC (Action de Recherche Concert\'{e}e) on Mining and Optimization of Big Data Models, funded by the Wallonia-Brussels Federation.
MTS acknowledges support from the European Union's Horizon 2020 research and innovation programme under the Marie Sklodowska-Curie grant agreement No 702410.

\bibliographystyle{siam}
\bibliography{ref}

\end{document}